\documentclass[a4paper,USenglish,thm-restate]{lipics-v2021}

\hideLIPIcs
\nolinenumbers

\usepackage[utf8]{inputenc}
\usepackage{graphicx} 
\usepackage{amsmath,amsthm}
\usepackage{algorithm,float,algorithmicx}
\usepackage[noend]{algpseudocode}
\usepackage[dvipsnames]{xcolor}
\usepackage{ifthen}

\usepackage{tikz}
\usetikzlibrary{positioning,fit}

\newcommand{\m}[1]{\ensuremath{\mathcal{#1}}}

\newcommand{\AP}{\mathsf{AP}}
\newcommand{\sat}{\models}
\newcommand{\set}[1]{\left\{ #1 \right\}}
\newcommand{\true}{\mathit{true}}
\newcommand{\false}{\mathit{false}}
\newcommand{\APop}{\mathsf{A}_{\mathit{op}}}
\newcommand{\AParg}{\mathsf{A}_{\mathit{arg}}}
\newcommand{\Met}{\mathit{Met}}
\newcommand{\Arg}{\mathit{Arg}}
\newcommand{\Ret}{\mathit{Ret}}

\NewDocumentCommand{\K}{O{}O{}}{%
  \ifx\relax#1\relax
    \ifx\relax#2\relax
      \mathsf{K}%
    \else
      \mathsf{K} \medspace #2%
    \fi
  \else
    \ifx\relax#2\relax
      \mathsf{K}_{#1}%
    \else
      \mathsf{K}_{#1} \medspace #2%
    \fi
  \fi
}

\newcommand{\APin}{\ensuremath{\mathsf{A}_\mathit{input}}}

\newcommand{\range}[3]{\ensuremath{[#1/2^{#3},#2/2^{#3}]}}

\title{Non-Leaking Concurrent Objects}

\author{Hagit Attiya}{Technion, Haifa, Israel}{hagit@cs.technion.ac.il}{https://orcid.org/0000-0002-8017-6457}{Supported by the Israel Science Foundation (22/1425 and 25/1849).}
\author{Rotem Oshman}{Tel-Aviv University, Tel-Aviv, Israel}{roshman@tau.ac.il}{https://orcid.org/0009-0007-5065-5557}{Supported by NSF-BSF 2022699.}
\author{Noa Schiller}{Tel-Aviv University, Tel-Aviv, Israel}{noaschiller@mail.tau.ac.il}{https://orcid.org/0009-0007-0285-6194}{Supported by NSF-BSF 2022699.}
\author{Corentin Travers}{Aix Marseille Univ, CNRS, LIS, Marseille, France}{corentin.travers@univ-amu.fr}{https://orcid.org/0000-0002-6797-4542}{Supported by ANR-23-PECL-0009 TRUSTINCloudS.}
\authorrunning{H. Attiya, R. Oshman, N. Schiller, and C. Travers}



\begin{document}
\maketitle

\begin{abstract}
Abstract specifications of concurrent objects determine which values operations may return, but they also implicitly constrain which information operations may know,
for example the arguments of other operations that do not affect their outcome, or even whether such operations occurred. 
Concrete implementations, while correct with respect to the abstract specification, may nonetheless expose additional information through their internal coordination mechanisms.

We introduce a framework for reasoning about information leakage in concurrent implementations. 
The framework uses epistemic logic to compare what a process may know under an abstract specification with what it may know in a concrete implementation, using the abstract object itself as the reference for permissible observations.
This yields several notions of \emph{non-leaking} implementations.

Using this framework, we investigate both the possibilities and limitations of non-leaking implementations. 
We present fully-non-leaking wait-free implementations of multi-valued registers and bounded max registers, but show that a fully-non-leaking unbounded max register cannot be implemented in a wait-free manner from finite-state base objects. 
We then consider a weaker guarantee, obtaining argument-non-leaking implementations of stacks, queues, and approximate agreement.
These results demonstrate that non-leakage guarantees are often compatible with correctness and progress requirements, while also indicating their limitations.
\end{abstract}

{\tableofcontents}


\section{Introduction}

Abstract specifications of concurrent objects describe the observable behavior of operations: which values operations may return and how these values relate to one another.
At the same time, such specifications implicitly constrain what an operation \emph{does not observe}.
For example, reading a variable in the abstract object does not reveal which values were written and later overwritten; popping an element from a stack does not reveal which other elements were concurrently pushed.
In this sense, abstract objects specify not only the information that must be preserved by an implementation, but also information that is intentionally left hidden.
Concrete implementations, even when correct with respect to the abstract behavior, may nonetheless expose this 
additional information through their internal coordination or contention mechanisms.

To see how the abstract specification determines which information an operation is allowed to infer, consider two examples. 
In a \emph{counter} object supporting \textsc{fetch\&inc}, 
if one operation returns $5$ and a later operation returns $8$, then in every execution of the abstract counter there are exactly two intervening \textsc{fetch\&inc} operations returning $6$ and $7$.
Thus, it is legitimate for a concrete implementation to reveal the existence of these intervening operations, since it is already implied by the abstract behavior.
In contrast, in a \emph{max register}, 
if one \textsc{readMax} returns $5$ and a later \textsc{readMax} returns $8$, the abstract specification does not imply that any intermediate \textsc{writeMax} operations occurred; the increase may be due to a single \textsc{writeMax}$(8)$.
Revealing the existence or arguments of such intermediate \textsc{writeMax} 
would therefore expose information that is not mandated by the abstract object.
For example, if a leaky max register is used for a \emph{highest-bidder wins} auction, 
it can inadvertently reveal the price offers of non-winning bids, which does not happen with the abstract max register.
These examples illustrate how the abstract specification itself determines which information may or may not be observable in a correct implementation.

This paper initiates the study of concurrent implementations that preserve these information-hiding guarantees. 
To formalize such guarantees, we introduce a semantic framework based on \emph{epistemic logic}, a standard formalism for reasoning about knowledge~\cite{FaginHMV2004}. 
At a high level, the framework compares what a process may learn under an abstract specification with what it may learn in a concrete implementation. 
The key idea is that the abstract object itself serves as the reference for permissible information exposure
and the implementation should not reveal more than that. 
That is, a non-leaking implementation does not reveal information about other operations beyond what is logically implied by the abstract specification and by the operation's own return value.
Returning to the examples above, the existence of intermediate \textsc{fetch\&inc} operations is implied by the abstract counter semantics, whereas the existence of intermediate \textsc{writeMax} operations is not.

In different situations we might want---or be able---to hide different types of information.
For example, a process may learn that an operation occurred, but not which arguments it carried.
To capture that, the framework is parameterized by the type of knowledge under consideration. 
In particular, we define two notions of \emph{non-leaking} implementations. 
In a \emph{fully-non-leaking} implementation, processes learn nothing about individual operations beyond what is implied by the abstract specification. 
In particular, they do not learn whether other operations occurred, nor do they learn their arguments or return values. 
A weaker notion, \emph{argument-non-leaking}, permits processes to learn that operations have occurred while preventing them from learning the identity of the invoking process or the arguments and return values of the operation. 
More generally, the framework allows one to express and reason about a spectrum of information-hiding guarantees for concurrent objects.

One of the central purposes of abstraction is to hide implementation details and allow clients to reason about an object solely through its specification. 
When an implementation reveals information that is not determined by the abstract object, it breaks this abstraction boundary: 
a client may distinguish between behaviors that are equivalent at the abstract level, 
and may know information through implementation mechanisms intended to deal with concurrency and contention.
A non-leaking implementation preserves the information-hiding guarantees implicit in the abstract specification by ensuring that operations do not acquire knowledge beyond what the abstraction itself exposes.

We investigate possibilities and limitations of information-hiding implementations. 
We present fully-non-leaking wait-free implementations of \emph{multi-valued registers} and \emph{bounded max registers}, but show that an \emph{unbounded} fully-non-leaking max register cannot be implemented in a wait-free manner from finite-state base objects. 
We then consider the weaker guarantee, and present an argument-non-leaking wait-free stack and lock-free queue. 
We also extend the framework to decision tasks, obtaining an argument-non-leaking wait-free approximate agreement algorithm. 
These results show that information-hiding guarantees can often be achieved without sacrificing standard correctness and progress properties, while also identifying inherent limitations of the strongest form of information hiding.

\subparagraph*{Related Work.} 

Traditional correctness notions for concurrent implementations 
are based on \emph{refinement}~\cite{Lamport86},
and they characterize \emph{admissible executions and return values}.
They do not constrain the information that processes can observe:
for example, a \emph{linearizable} implementation~\cite{HerlihyWing90}
may expose internal state or intermediate results
not determined by the abstract specification.
Leaking this information to the \emph{scheduling adversary}
leads to the surprising behaviors observed in~\cite{GolabHW2011} (which are addressed by imposing the stronger requirement of \emph{strong linearizability}, preventing the adversary from being able to ``change the linearization'' after learning internal information).
Our work goes in a different direction:
our goal is to protect intermediate states from observation by \emph{other processes},
so that they do not learn specific information not revealed by the abstract specification.

Our epistemic framework for information gained by processes
is based on \emph{indistinguishability} of executions from the perspective of a single process, a standard technique in distributed computing~\cite{AttiyaR2020}. 
Recent work on \emph{auditing under curiosity}~\cite{AuditingPODC2025} also employs
indistinguishability-based reasoning to limit information disclosure, but is 
restricted to audit and logging mechanisms. 
Our approach instead uses the abstract object specification as the reference for allowable information exposure for \emph{arbitrary} objects.

Another closely related work studies \emph{concurrent history
independence}~\cite{HIPODC2024,HISTOC2025}, where an external observer may
inspect the entire memory state at quiescent points and should not infer
information beyond what is implied by the current state. Under this notion, even
a single-writer $N$-valued register cannot be implemented from base objects with
fewer than $N$ states~\cite{HIPODC2024}. 
In contrast, we present 
a fully-non-leaking implementation 
of such a register from binary registers.
This separates non-leaking and history-independent implementations: 
the former reason about what an executing operation can infer during execution, 
while the latter considers what an external observer can infer from quiescent memory states.

In \emph{amnesic} register implementations~\cite{ChocklerGuKe07}, 
old values are eventually deleted from memory.
This definition takes a perspective similar to history independence, as it considers the state of the entire memory configuration rather than the execution of a particular operation.

There is extensive work on information hiding in other settings.
\emph{History independence} studies what information about past operations can
be inferred from the memory representation of a data structure~\cite{NaorTeague01},
while \emph{information-flow security} studies how secret inputs may influence
observable outputs~\cite{GoguenMeseguer82}. 
Related notions of observational or
contextual equivalence and refinement appear in programming languages, where
they are used to reason about substitutability~\cite{Milner1999}.
These approaches typically consider external observers, final memory states, or
explicit security policies. In contrast, we study what a process can
infer from its own observations while executing operations, using
the abstract object specification as the reference for permissible information
exposure.

We also note a connection to \emph{differential privacy}~\cite{DworkR2014}.
Our setting differs in that the observer is an operation executing the implementation (rather than an external analyst), and the relevant changes are guided by the abstract object semantics.
This perspective suggests quantitative variants of non-leaking properties, obtained by requiring the distribution of an operation's view to change only slightly under such semantic neighbors.
Non-leaking decision tasks are reminiscent of simulation-based definitions in secure multi-party computation~\cite{Yao1982},
though our setting is non-cryptographic since we assume a weaker adversary.

\section{Definitions}  
\label{sec:definitions}

\subsection{Labeled Transition Systems and Refinement}

We use labeled transition systems~\cite{Keller1976} to model 
both abstract and concrete objects (specifications and implementations, respectively).
A \emph{labeled transition system} (\emph{LTS}) 
$(Q, \Sigma,  s_0, \delta)$ over the possibly-infinite alphabet
$\Sigma$ is a possibly-infinite set $Q$ of states with
initial state $s_0 \in Q$, 
and a transition relation $\delta \subseteq Q \times \Sigma \times Q$.
The $i$th symbol of a sequence $\tau \in \Sigma^*$ is denoted $\tau_i$, 
and $\epsilon$ is the empty sequence.
An LTS is \emph{deterministic} if for any state $s$ 
and any sequence $\tau\in \Sigma^*$, there is at most
one state $s'$ such that $s\xrightarrow{\tau}s'$. 

 A finite \emph{execution} of a LTS is an alternating sequence 
of states and transition labels (also called \emph{actions})
$\rho = s_0, a_0,s_1\ldots a_{k-1},s_k$ for some $k>0$ 
such that $(s_i, a_i, s_{i+1})\in \delta$ for each $0\leq i<k$. 
The projection $\tau| \Gamma$ of a sequence $\tau$ is the maximum 
subsequence of $\tau$ over alphabet $\Gamma \subseteq \Sigma$;
this extends to sets of sequences as usual.
A \emph{trace} is the projection $\rho | \Sigma$ of an execution $\rho$.
$\mathit{E}(A)$ is the set of executions of an LTS $A$,
while $\mathit{T}(A)$ is the set of its traces.

Throughout the paper, we assume a fixed set of processes $P$.
An \emph{object} is a \emph{deterministic} LTS over alphabet 
$\textit{Call} \cup \textit{Ret} \cup \Sigma_o$ where $\textit{Call}$, 
resp., $\textit{Ret}$, is the set of call, resp., return, actions, 
and $\Sigma_o$ is an alphabet of internal actions. 
Formally, a call action $call(m,d,p,k)$, resp., a return action $ret(m,d,p,k)$, 
combines a method $m$ and argument, resp., return value, 
$d$ with a process identifier $p$ and an operation identifier $k$. 
Operation identifiers are used to pair call and return actions.
We assume that the traces of an object satisfy standard well-formedness properties, 
e.g., return actions correspond to previous call actions.

A \emph{configuration} of the system contains the states of
all shared base objects and processes.
In an \emph{initial} configuration,
base objects and processes are in their initial states.

Given a standard description of an object implementation as a set of methods, 
its LTS represents the executions of its most general client 
(that may call methods in any order and from any thread). 
It is convenient to conflate the \emph{states} of the LTS with its \emph{executions} by assuming w.l.o.g.\ that the state space is simply the set of executions.
Because we consider only deterministic implementations, each execution ends in a unique configuration, so each state of the LTS corresponds to a unique configuration of the system.
The transitions of the LTS correspond to statements in the method bodies 
(in which case they are labeled by internal actions in $\Sigma_o$), 
or call and return actions. 
A trace $\tau$ of an object $O$ projected over call and return actions 
is called a \emph{history} of $O$, and it is denoted by $\mathit{hist}(\tau)$.
The set of histories admitted by an object $O$ is denoted by $H(O)$.

Call and return actions $call(m,\_,p,k)$ and $ret(m,\_,p,k)$ are called \emph{matching} 
when they contain the same process and operation identifiers $p,k$. 
A call action is \emph{unmatched} in a history $h$ 
when $h$ does not contain the matching return.
We often refer to a call action $call(m,v_1,p,k)$ and its matching return action
$ret(m,v_2,p,k)$, if it exists, 
as an \emph{operation} with argument $v_1$ and return value $v_2$,
denoted $op(m,p,v_1,v_2)$.
A history $h$ is \emph{sequential} if every call $call(m,\_,p,k)$ 
is immediately   followed by the matching return $ret(m,\_,p,k)$. 

The \emph{view} of a process $p$ in an execution $\gamma$ is the projection of $\gamma$ onto actions taken by process $p$ (i.e., operation calls and returns where the process identifier is $p$, and internal actions taken by process $p$).

An object $C$ \emph{refines} object $A$
if all histories of $C$ are histories of $A$,
i.e., $H(C) \subseteq H(A)$.
In this case we often refer to $C$ as the \emph{concrete object}, and to $A$ as the \emph{abstract object} that $C$ refines.
A \emph{refinement mapping from $C$ to $A$} is a mapping $r : H(C) \rightarrow H(A)$;
clearly, a refinement mapping from $C$ to $A$ exists iff $C$ refines $A$.

In this paper, the concrete objects are implementations of abstract 
objects in the standard asynchronous shared-memory model~\cite{AttiyaW2004} 
with $n$ processes, $P = \set{ p_0,\ldots,p_{n-1} }$.
Processes communicate with each other by applying \emph{primitive} 
operations to shared \emph{base objects}, 
and these are the actions of the concrete object. 
A concrete \emph{implementation} $C$ of and abstract object $A$ 
is an algorithm that provides a local code for each process $p$,
which specifies the primitive operations on base objects $p$ applies
and local computations $p$ performs in order to return a response
when it invokes an operation of $A$.

\emph{Linearizability}~\cite{HerlihyWing90} is essentially refinement
when the abstract object $A$ is sequential. 
Intuitively, $C$ is a \emph{linearizable} implementation of $A$ if each operation
appears to take effect \emph{atomically} at some time between 
its invocation and response, 
hence operations' real-time order is maintained. 
We consider the standard \emph{wait-free} and \emph{lock-free} 
progress conditions.

\subparagraph*{Indistinguishable executions and ``no forgetting''.}
Two executions 
$\gamma$ and $\gamma'$ are \emph{indistinguishable} 
to process $p$, denoted $\gamma \sim_p \gamma'$,
if it has the same view in both executions.
This will be important in the context of our definition of non-leaking implementations:
it means that processes can use past observations 
to try to extract information about other processes.

\subsection{Non-Leaking in Terms of Epistemic Logic}

Let $L = (\Sigma, S, \delta)$ be an LTS representing 
a concurrent object for a set of processes $P$. 
In addition, let $\AP$ be a set of atomic propositions,
and let 
$v : S \times \AP \rightarrow \set{ \true, \false}$
be an evaluation function
indicating for each state $s \in S$ and atomic proposition $p$ whether $p$ holds true in state $s$.
We denote $s \sat p$ if $v(s, p) = \true$.

Syntactically, the formulas of \emph{epistemic logic} are inductively 
defined by the base set $\AP$ and the operators 
$\set{ f_{\land}(\cdot,\cdot), f_{\lor}(\cdot, \cdot), f_{\lnot}(\cdot)} \cup \set{ f_{\K[p][]}(\cdot) : p \in P}$,
where $f_{\circ}(\alpha,\beta) = (\alpha \circ \beta)$ for $\circ \in \set{\land, \lor}$
and $f_{\circ}(\alpha) = \circ (\alpha)$ for $\circ \in \set{ \lnot} \cup \set{\K[p][] : p \in P}$. We omit the parentheses when the order of operations is clear, 
with $\lnot$ and $\K[p][]$ taking precedence over $\land,\lor$.

The semantics of epistemic logic is also defined inductively, 
with the base case (atomic formulas) already defined above,
and the Boolean operators $\land, \lor, \lnot$ evaluated as usual.
The knowledge operator $\K[p][]$
is evaluated as follows:
$s \models \K[p][\varphi]$ iff for every state $s' \in S$
such that $s' \sim_p s$,
we have $s' \models \varphi$. (Note that this includes state $s$ itself. Here, $\sim_p$ is the indistinguishability relation defined above.)

Recall that we conflate the \emph{states} of an LTS representing 
a concurrent object with the \emph{executions}
of the object, so we will frequently write expressions such as $\beta \models \varphi$ to indicate that formula $\varphi$ holds in the state of the LTS that corresponds to execution $\beta$.

Let $L_{abs} = (\Sigma_{abs}, S_{abs}, \delta_{abs}), L_{conc} = (\Sigma_{conc}, S_{conc}, \delta_{conc})$ be two LTSs enriched with the same set of atomic propositions $\AP$.
Fix a set of atomic propositions $A \subseteq \AP$.

If $L_{conc}$ refines $L_{abs}$ (as defined above), 
then we say that it is a \emph{$A$-non-leaking implementation of $L_{abs}$} if for every $a \in A$, process $p$,
and executions $\alpha$ and $\beta$ of $L_{abs}$ and $L_{conc}$ (resp.) such that $r(\beta) = \alpha$
and process $p$ has no pending operations in $\beta$,
    \begin{equation*}
    \text{If}
    \medspace
    \alpha \not \models \K[p][a]
    \medspace
    \text{then}
    \medspace
    \beta \not \models \K[p][a]
    .
    \end{equation*}

\begin{proposition}
    Let $A_1 \subseteq A_2$ be sets of atomic propositions.
    If an implementation is $A_2$-non-leaking, then it is also $A_1$-non-leaking.
\end{proposition}

The framework that we introduced can capture non-leaking of \emph{any} class of facts (i.e., atomic propositions), but in this paper we focus on leakage of individual operations performed by other processes, and/or the arguments and return values of those operations, defined next.

\subsection{Non-Leaking of Operations of the Abstract Object}
Let $\Met$ be the set of methods supported by the object,
and for each $m \in \Met$,
let $\Arg_m,\Ret_m$ be the possible arguments and return values of method $m$ (resp.).
We slightly abuse notation by letting $op(m, p, v_1, v_2)$ be an atomic proposition
which holds true in state $\gamma$ of the LTS
if and only if 
$op(m, p, v_1, v_2)$ appears (and is complete) in the execution $\gamma$ (for $m \in \Met, p \in P, v_1 \in \Arg_m, v_2 \in \Ret_m$).

We define two sets of atomic propositions,
each reflecting a different type of knowledge about operations
that have occurred in the execution:
\begin{itemize}
    \item Knowledge that a specific operation was invoked:
    \begin{equation*}
        \APop = \set{ 
        \medspace
        \exists p \in P \medspace
        \exists a \in \Arg_m \medspace
        \exists r \in \Ret_m
        \medspace.\medspace
        op(m, p, a, r) 
        \quad : \quad
        m \in \Met
        \medspace
        }
        .
    \end{equation*}

    \item Knowledge about the invoking process, arguments and/or return value:
         \begin{align*}
        \AParg &= \set{ 
        \medspace
\exists (p, a, r) \in V 
        \medspace.\medspace
        op(m, p, a, r)        \quad : \quad 
        m \in \Met,
        V \subsetneq P \times \Arg_m \times \Ret_m
        \medspace
             }.
    \end{align*}

\end{itemize}

   An implementation 
    is called \emph{fully-non-leaking}
    if it is $\left( \APop \cup \AParg \right)$-non-leaking,
    and is \emph{argument-non-leaking}
    if it is $\AParg$-non-leaking.

Given a set $V \subseteq P \times \Arg_m \times \Ret_m$,
we use the following short-hand notation. Let
    \begin{equation*}
        op(m, V) = \exists (p, a, r) \in V 
        \medspace.\medspace
        op(m, p, a, r) 
    \end{equation*}
    be an atomic proposition asserting
    that an operation $op(m,p,a,r)$
    with $(p,a,r) \in V$
    has occurred.
    Using this notation we have
    $\APop = \set{ op(m, P \times \Arg_m \times \Ret_m ) : m \in \Met}$,
    and 
    $\AParg = \set{ op(m, V ) : m \in \Met,  V \subsetneq P \times \Arg_m \times \Ret_m }$.
    If an operation has no arguments and/or no return value, we may omit them from our notation.

\subparagraph*{A useful proof technique.}
Recall that the assertion $\gamma \not \models \K[p][a]$ means that there exists a $p$-indistinguishable execution $\gamma'$ such that $\gamma' \not\models a$.
Thus, proving that an implementation is $A$-non-leaking boils down to the following:
given an atomic proposition $a \in A$,
abstract executions $\alpha, \alpha'$ such that $\alpha \sim_p \alpha'$ and $\alpha' \not\models a$,
and a concrete execution $\beta$ that refines $\alpha$,
we must exhibit a concrete execution $\beta'$ such that $\beta \sim_p \beta'$ and $\beta' \not\models a$ (see Fig.~\ref{fig:abstract-concrete}).

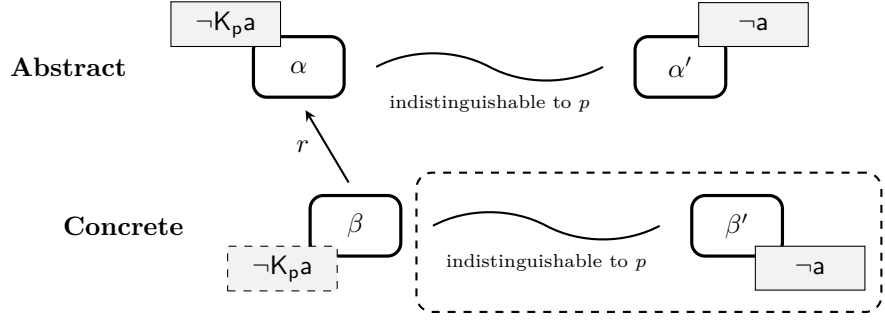
\begin{figure}[t]
\centering
\begin{tikzpicture}[
    state/.style={
        draw,
        very thick,
        rounded corners=5pt,
        minimum width=1.2cm,
        minimum height=0.8cm,
        inner sep=3pt
    },
    predtag/.style={
        draw,
        thin,
        rectangle,
        fill=gray!10,
        inner sep=3.5pt,
        minimum width=1.45cm,
        minimum height=0.58cm,
        align=center
    },
    prooftag/.style={
        draw,
        thin,
        dashed,
        rectangle,
        fill=gray!10,
        inner sep=3.5pt,
        minimum width=1.45cm,
        minimum height=0.58cm,
        align=center
    },
    relationbox/.style={
        minimum width=3.0cm,
        minimum height=0.50cm,
        inner sep=0pt
    },
    relationcurve/.style={
        line width=0.8pt,
        line cap=butt
    },
    relationtext/.style={
        font=\scriptsize,
        inner sep=0pt
    },
    arr/.style={
        ->,
        >=stealth,
        thick,
        shorten <=5pt,
        shorten >=5pt
    }
]

\def\abstractY{2.10}

\node[anchor=east] at (-2.15,\abstractY)
    {\textbf{Abstract}};

\node[anchor=east] at (-1.40,0)
    {\textbf{Concrete}};

\node[state] (alpha) at (0,\abstractY)
    {$\alpha$};

\node[state] (alphap) at (5.05,\abstractY)
    {$\alpha'$};

\node[state] (beta) at (0.75,0)
    {$\beta$};

\node[state] (betap) at (5.80,0)
    {$\beta'$};

\node[predtag, anchor=south east] (tagalpha)
    at ([xshift=11pt,yshift=-4.5pt]alpha.north west)
    {$\mathsf{\neg K_p a}$};

\node[predtag, anchor=south west] (tagalphap)
    at ([xshift=-11pt,yshift=-4.5pt]alphap.north east)
    {$\mathsf{\neg a}$};

\node[prooftag, anchor=north east] (tagbeta)
    at ([xshift=11pt,yshift=4.5pt]beta.south west)
    {$\mathsf{\neg K_p a}$};

\node[predtag, anchor=north west] (tagbetap)
    at ([xshift=-11pt,yshift=4.5pt]betap.south east)
    {$\mathsf{\neg a}$};

\node[relationbox] (relalpha)
    at (2.525,\abstractY) {};

\node[relationbox] (relbeta)
    at (3.275,0) {};

\draw[relationcurve]
    ([xshift=-1.50cm]relalpha.center)
    .. controls
        ([xshift=-1.05cm,yshift=0.24cm]relalpha.center)
        and
        ([xshift=-0.45cm,yshift=0.24cm]relalpha.center)
    ..
    (relalpha.center)
    .. controls
        ([xshift=0.45cm,yshift=-0.24cm]relalpha.center)
        and
        ([xshift=1.05cm,yshift=-0.24cm]relalpha.center)
    ..
    ([xshift=1.50cm]relalpha.center);

\node[relationtext, below=3pt of relalpha] (relalphatext)
    {indistinguishable to \(p\)};

\draw[relationcurve]
    ([xshift=-1.50cm]relbeta.center)
    .. controls
        ([xshift=-1.05cm,yshift=0.24cm]relbeta.center)
        and
        ([xshift=-0.45cm,yshift=0.24cm]relbeta.center)
    ..
    (relbeta.center)
    .. controls
        ([xshift=0.45cm,yshift=-0.24cm]relbeta.center)
        and
        ([xshift=1.05cm,yshift=-0.24cm]relbeta.center)
    ..
    ([xshift=1.50cm]relbeta.center);

\node[relationtext, below=3pt of relbeta] (relbetatext)
    {indistinguishable to \(p\)};

\draw[arr]
    (beta.north) --
    node[left=3pt] {$r$}
    (alpha.south);

\node[
    draw,
    dashed,
    thick,
    rounded corners=6pt,
    fit={
        (relbeta)
        (relbetatext)
        (betap)
        (tagbetap)
    },
    inner xsep=6pt,
    inner ysep=8pt
] {};

\end{tikzpicture}
\caption{Proving that an implementation is non-leaking with respect to atomic proposition $a$. Dashed boxes indicate claims that must be established: namely, to prove that $\beta \not \models \K[p][a]$,
we must exhibit a $p$-indistinguishable execution $\beta'$ such that $\beta' \not \models a$.}
\label{fig:abstract-concrete}
\end{figure}

To find $\beta'$, it is helpful to consider a ``simplest possible'' abstract execution 
$\alpha'' \sim_p \alpha'$ such that $\alpha'' \not\models a$,
and try to find a concrete execution $\beta'$ that ``matches it''.
In our case, since the atomic propositions that we consider deal with the presence of operations in the trace, ``simplest'' translates to having as few operations as possible.
Thus, for an abstract execution $\gamma$,
let us define $\min_p(\gamma)$ to be 
a ``minimal $p$-indistinguishable version of $\gamma$'':
$\min_p(\gamma)$ contains a subset of the operations in $\gamma$,
such that $\min_p(\gamma) \sim_p \gamma$,
but we cannot remove any operation from $\min_p(\gamma)$
without violating $p$-indistinguishability.
(There may be more than one such execution, in which case we fix one arbitrarily.)
Working with this minimal execution allows us to prove non-leakage without considering each atomic proposition individually,
by proving the following relationship between abstract and concrete executions:

    \begin{lemma}
    \label{lem:fnl reduction}
    Let $L_{conc}$ be an implementation of $L_{abs}$ with a refinement mapping $r$, and assume that for every process $p$,
    every pair of abstract executions $\alpha, \alpha'$ (of $L_{abs}$) such that $\alpha' \sim_p \alpha$,
    and every concrete execution $\beta$ (of $L_{conc}$) such that $r(\beta) = \alpha$,
     there exists a concrete execution $\beta'$ (of $L_{conc}$)
    such that $\beta' \sim_p \beta$
    and $\beta'$ contains exactly the operations in $\min_p(\alpha')$.
    Then $L_{conc}$ is a fully-non-leaking implementation of $L_{abs}$.
    \end{lemma}

    \begin{proof}
    Let $a$ be a proposition in $\APop \cup \AParg$ such that 
    $\alpha \not \models \K[p][a]$. 
    We show that $\beta \not \models \K[p][a]$.

    Since $\alpha \not \models \K[p][a]$, there is an abstract execution $\alpha' \sim_p \alpha$ such that $\alpha' \not \models a$. Since the atomic proposition $a$ concerns the existence of a specific operation, $\gamma \not \models a$ for any (concrete or abstract) execution $\gamma$ containing only a subset of the operations in $\alpha'$. Hence, it also holds that $\min_p(\alpha') \not \models a$.      
    By assumption, 
    there is an execution $\beta' \sim_p \beta$ containing exactly the operations in $\min_p(\alpha')$.
    This implies that $\beta' \not \models a$, and therefore, $\beta \not \models \K[p][a]$.
    \end{proof} 

\subparagraph*{Discussion.}
The definition of $\AParg$ is more subtle than $\APop$,
and is intended to capture \emph{any} knowledge gained about the invoking process and the arguments and return value of the operation, including relationships between them.
For example, it captures the following facts:
\begin{itemize}
    \item Operation $m$ was invoked by process $p$ with a specific argument $a$ and return value $r$:
    we take
    $V = \set{ p } \times \set{ a } \times \set{ r} $.
    \item For an operation $m$ that takes one argument and returns a value from the same domain
    ($\Arg_m = \Ret_m$),
    the fact that 
    operation $m$ was invoked by an unknown process, and its return value was the same as its argument:
    $V = P \times \set{ (v, v) : v \in \Arg_m }$.
    (Note that we do not learn what the argument or return value \emph{are},
    only that they are the same.)
\end{itemize}
In order to be $\AParg$-non-leaking an implementation must not leak \emph{any} such fact.

We also remark that $\APop$ and $\AParg$ are incomparable in the sense that an implementation may be $\APop$-non-leaking but not $\AParg$-non-leaking and vice-versa.
For example, an implementation of a register may not leak that any particular operation has occurred (and so is $\APop$-non-leaking), but in situations where the abstract specification \emph{allows} learning that a specific operation has occurred,
we may learn the arguments of the operation even though the abstract specification does not reveal them (and so the implementation is not $\AParg$-non-leaking).
In the other direction, an implementation may leak that certain operations have occurred (and thus not be $\APop$-non-leaking) while not revealing anything about the arguments or return values (and thus be $\AParg$-non-leaking); this is the case for the stack and queue implementations described in Section~\ref{sec:stack_queue}.

\section{A Wait-Free Fully-Non-Leaking $N$-Valued Single-Writer Register}

We start with a simple and elegant fully-non-leaking implementation 
of a single-writer multi-reader $N$-valued register, 
from binary registers.
Algorithm~\ref{alg:reg} is a recursive implementation of an
$N$-valued register,
from an $\lceil N/2\rceil$-valued register $r$ 
and an array $A$ of $\lceil N/2\rceil$ binary registers.
Recursively implementing $r$ allows using only binary registers. 

Each value $x$ of $r$ represents the two values $2x$ and $2x + 1$; the corresponding entry $A[x]$ indicates the parity of the most recent value written among the two.
We assume that $r$ and $A[0]$ are initialized to 0, representing that the register value is initially 0.
A \textsc{write}$(x)$ first writes to $A[\lfloor x/2\rfloor]$ the parity of $x$, and then announces the new value by writing $\lfloor x/2\rfloor$ to $r$.
A \textsc{read} operation accesses the memory in the reverse order of the writes: it first reads a value $x$ from $r$ and a parity $p$ from $A[x]$, returning the final value $2x + p$.

\begin{algorithm}[bt]
  \caption{Fully-non-leaking wait-free single-writer $N$-valued register}
  \label{alg:reg}
  \begin{algorithmic}[1]
  \small
    \State\textbf{Shared objects}
    \State\hspace{\algorithmicindent}$r$: $\lceil N/2\rceil$-valued register, initialized to $0$
    \State\hspace{\algorithmicindent}$A$: array $[0..\lceil N/2\rceil-1]$ of binary registers (each initially 0)
    \Statex
    \begin{minipage}[t]{0.45\textwidth}
    \Function{write}{$x$}
      \State $A[\lfloor x/2\rfloor].\mathsf{write}(x \bmod 2)$
      \State $r.\mathsf{write}(\lfloor x/2\rfloor)$
    \EndFunction
    \end{minipage}
    \begin{minipage}[t]{0.45\textwidth}
    \Function{read}{$~$}
      \State $x \gets r.\mathsf{read}()$
      \State \Return{$2x + A[x].\mathsf{read}()$}
      \EndFunction
      \end{minipage}
  \end{algorithmic}
\end{algorithm}

Unfolding the recursive construction leads to a binary tree of binary registers with height $\lceil \log_2 N \rceil$.
The writer updates the value bit-by-bit, starting from the least significant bit in the leaves and concluding with the most significant bit at the root.

We linearize the operations as follows.
We begin with a $\textsc{write}(x)$ operation $wop$. The linearization point for $wop$ is defined as the earliest step in the execution between $wop$ performing $r.\mathsf{write}(\lfloor x/2 \rfloor)$ and any step in which a $\textsc{read}$ operation reads the value $wop$ writes to $A[\lfloor x/2 \rfloor]$.
To linearize \textsc{read} operations we prove 
(Lemma~\ref{lem:reg read from})
that if a \textsc{read} operation reads the value $b$ from $A[x]$, 
then either $b=0$ and $x = 0$, or this value was written 
by a \textsc{write}$(2x + b)$ operation.
A $\textsc{read}$ operation $rop$ is linearized after 
the $\textsc{write}$ operation that writes the value to $A$ which $rop$ reads.
If $rop$ reads the initial value from $A$, 
it is linearized at the beginning before any $\textsc{write}$ operation.

\begin{lemma}
\label{lem:reg read from}
    If a \textsc{read} operation reads the value $b$ from $A[x]$, then either $b=0$ and $x = 0$, or this value was written by a \textsc{write}$(2x + b)$ operation.
\end{lemma}

\begin{proof}
    Before a $\textsc{read}$ operation $rop$ reads $A[x]$, it first reads the value $x$ from $r$. Conversely, before a $\textsc{write}(x)$ operation writes the value $\lfloor x/2 \rfloor$ to $r$, it first writes to $A[\lfloor x/2 \rfloor]$.
    
    Consequently, $rop$ either reads the initial values from $r$ and $A[0]$, or it reads a value $x$ that was written by some $\textsc{write}(2x)$ or $\textsc{write}(2x+1)$ operation. As explained above, the ordering of these steps ensures that by the time $x$ is visible in $r$, the value $b$ read from $A[x]$ is no longer the initial value, but rather the value stored there by a $\textsc{write}(2x+b)$ operation.
\end{proof}

\begin{lemma}
\label{lem:reg linearize}
Algorithm~\ref{alg:reg} implements a linearizable single-writer $N$-valued register.
\end{lemma}

\begin{proof}
    Lemma~\ref{lem:reg read from} shows that the linearization satisfies the sequential specification of a register.
    We show that the linearization preserves the real-time order of non-overlapping operations.
    The order between \textsc{write} operations is preserved because they are linearized at a step between the invocation and response of the operation.

    Consider a $\textsc{read}$ operation $rop$. $rop$ cannot precede the $\textsc{write}$ operation $wop$ that it is linearized after, as $rop$ reads a value written by $wop$. Because all other $\textsc{write}$ operations placed before $rop$ in the linearization precede $wop$, it follows that $rop$ cannot precede them.

    Consider a \textsc{write} operation $wop$ that precedes $rop$. Then, $rop$ reads from $r$ a value written by $wop'$ where $wop' = wop$ or $wop'$ follows $wop$ (recall all \textsc{write} operations are non-overlapping).
    Since $wop'$ first writes to $A$ before writing to $r$, it also follows that $rop$ reads from $A$ a value written by $wop''$ where $wop'' = wop'$ or $wop''$ follows $wop'$.
    $rop$ is linearized after $wop''$, and is therefore also linearized after $wop$.

    Consider two $\textsc{read}$ operations, $rop_1$ and $rop_2$, such that $rop_1$ precedes $rop_2$ in the execution.
    Let $wop$ be the \textsc{write} operation $rop_1$ is linearized after.

    Note that $wop$ writes to $A$ before $rop_1$ returns.
    Since \textsc{write} operations do not overlap, $rop_2$ reads from $r$ a value written by $wop'$ where $wop' = wop$ or $wop'$ follows $wop$. 
    Since $wop'$ first writes to $A$ before writing to $r$, it also follows that $rop_2$ reads from $A$ a value written by $wop''$ where $wop'' = wop'$ or $wop''$ follows $wop'$.
    $rop_2$ is linearized after $wop''$, and is therefore also linearized after $wop$.
    This implies $rop_2$ is also linearized after $rop_1$. 
\end{proof}

Next we prove that it is fully-non-leaking.
A \textsc{write} operation only writes to memory, so the 
writer does not learn anything about other operations.
The non-leakage property for \textsc{read} operations follows by observing that a $\textsc{read}$ operation returns $x$ if and only if it reads $\lfloor x/2 \rfloor$ from $r$ and $x \bmod 2$ from $A[\lfloor x/2 \rfloor]$.
These are the only two memory reads a $\textsc{read}$ operation performs, implying:

\begin{lemma}
\label{lem:non-leaking-reg}
Algorithm~\ref{alg:reg} is a fully-non-leaking implementation of a single-writer register.
\end{lemma}

\begin{proof}
    Let $\alpha,\beta$ be executions of respectively the abstract register and Algorithm~\ref{alg:reg} such that $\alpha$ is the linearization of $\beta$, as defined in the proof of Lemma~\ref{lem:reg linearize}.
    Let $\alpha' \sim_p \alpha$. 
    We show that there exists a concrete execution $\beta' \sim_p \beta$ that contains exactly the operations in $\min_p(\alpha')$.
    By Lemma~\ref{lem:fnl reduction}, this implies that the algorithm is fully-non-leaking.

  We construct $\beta'$ by iterating over the operations of $p$ in $\min_p(\alpha')$.
  Since $\min_p(\alpha') \sim_p \alpha$ and $\alpha$ is the linearization of $\beta$, both $\min_p(\alpha')$ and $\beta$ contain the exact same operations by process $p$, appearing in the exact same order. 
    If the next operation of $p$ is a \textsc{write} operation, we simply append it to $\beta'$. We can do this because $p$ only writes during the operation and never reads from memory.
    Thus, $p$ gains no knowledge about the other operations in $\beta$.

    If the next operation of $p$ is a \textsc{read} operation that returns $x$, then consider the last \textsc{write} operation (if it exists) in $\beta'$. 
    If the argument for this \textsc{write} is $x$, we simply append the \textsc{read} operation by $p$ to $\beta'$.
    Otherwise, since $\min_p(\alpha')$ respects the sequential specification, the preceding \textsc{write} operation in $\min_p(\alpha')$ has argument $x$ and is by process $q\neq p$.
    We first append a complete \textsc{write}$(x)$ by process $q$, and then append the \textsc{read} operation by $p$ to $\beta'$.

    By construction, $\beta' \sim_p \beta$. In addition, due to the minimality of $\min_p(\alpha')$, any operation not executed by $p$ exists solely to justify a value returned by a \textsc{read} operation by $p$. Hence, $\min_p(\alpha')$ does not contain any operation that is not included in $\beta'$.
\end{proof}

\section{A Wait-Free Fully-Non-Leaking Bounded Max Register}

A max register~\cite{AspnesAttiyaCensor2012} stores values from 
a totally ordered domain $V$ and supports two operations, 
\textsc{writeMax}$(x)$ and \textsc{readMax}(), 
where \textsc{readMax} returns the largest value
previously written by a \textsc{writeMax} operation.
In this section, we present a non-leaking implementation of a max register with finite domain $V = \{0,\ldots,N-1\}$.

Algorithm~\ref{alg:reg} is reminiscent of the bounded max register 
of~\cite{AspnesAttiyaCensor2012}, where the value is also written bit-by-bit starting from the least significant bit using the same tree structure.
Note, however, that their max register implementation is leaking, 
as writers condition their writes on values read at each level.
(For a reader familiar with this algorithm: 
if the current bit of the value being written is 0, but the bit at this level is 1, the writer terminates its execution, but learns that a larger value has been written. 
Omitting this check is difficult since it prevents writers of smaller values from overwriting the bits of larger values.)

Algorithm~\ref{alg:maxregister} is a fully-non-leaking bounded max register.
It uses two arrays of binary registers, $S[0..N-1]$ and $I[0..N-1]$. The value $1$ in $S[x]$ marks that $x$ is \textit{set}, while a $1$ in $I[x]$ marks that $x$ is \emph{invalidated}. A read operation returns $x$ only if it finds that $I[x]=0$ and $S[x]=1$. $S[0]$ is initially 1, 
representing that the max register is initially $0$.

A \textsc{writeMax}$(x)$ first writes $1$ to $S[x]$
then invalidates all markers $I[y]$ with $y<x$ by writing 1 to $I[y]$. 
This both records that some process wrote at $x$ and eliminates evidence of
smaller values. 
A \textsc{readMax} operation goes over the arrays of markers $I$, from $N-1$ downwards. 
If it finds an index $x$ where $I[x] = 0$ and $S[x] = 1$, 
it invalidates all $I[y]$ for $y<x$ and returns $x$. 
However, if the scan encounters an invalid marker ($I[x] = 1$), it aborts the downward search and initiates an upward scan starting from $x+1$. Upon finding the first index $y$ in this upward scan such that $I[y] = 0$, the operation similarly invalidates all preceding markers and returns $y$.

\begin{algorithm}[bt]
  \caption{Fully-non-leaking wait-free max register with input domain $\{0,\ldots,N-1\}$}
  \label{alg:maxregister}
  \begin{algorithmic}[1]
  \small
    \State\textbf{Shared objects}
    \State\hspace{\algorithmicindent}$S[0..N-1]$: array of $N$ registers, initially 0 except $S[0] = 1$
    \State\hspace{\algorithmicindent}$I[0..N-1]$: array of $N$ registers, initially 0
    
    \Function{writeMax}{$x$}
    \State $S[x].\mathsf{write}(1)$
    \For{$y = x-1, x-2,\ldots,0$} $I[y].\mathsf{write}(1)$
    \EndFor
    \EndFunction
    \Function{readMax}{$~$}
    \For{$x= N-1,N-2,\ldots , 0$}
    \Comment{Downward scan}
    \If{$I[x].\mathsf{read}() = 1$}
        \For{$z= x+1,\ldots ,N-1$}
        \Comment{Upward scan}
        \If{$I[z].\mathsf{read}() = 0$}
        \For{$y= z-1,z-2,\ldots, 0$} $I[y].\mathsf{write}(1)$
        \EndFor
        \State \Return{$z$}
    \EndIf
    \EndFor
    \EndIf
    \If{$S[x].\mathsf{read}() = 1$}
    \For{$y= x-1,x-2,\ldots, 0$} $I[y].\mathsf{write}(1)$
    \EndFor
    \State \Return{$x$}
    \EndIf
    \EndFor
    \EndFunction 
  \end{algorithmic}
\end{algorithm}

Termination of \textsc{writeMax} is immediate. 
If \textsc{readMax} finds a set value during its downward scan, it returns a valid value.
If it finds that $I[x] = 1$ during its downward scan, the subsequent upward scan (starting from $I[x+1]$) is guaranteed to find $0$ at some $I[y]$, $y > x$,
because the entry in $I$ corresponding to the largest value ever written cannot be invalidated. Since no larger value is ever written, neither this entry nor any entries corresponding to larger values can be invalidated.
This is formalized in the next proposition.

\begin{proposition}
\label{prop:max val set}
    Let $x$ be the largest value such that $S[x].\mathsf{write}(1)$ was performed during the execution. Then for all $y \geq x$, $I[y] = 0$.
\end{proposition}

A \textsc{writeMax}$(x)$ operation $wop$ is \emph{visible} if $wop$ changes $S[x]$ from $0$ to $1$.
We first order all visible \textsc{writeMax} operations by their arguments, from smallest to largest.
Note that since each $S[x]$ can be set to 1 only once, there is at most one visible \textsc{writeMax}($x$) for each $x$.

Next, we place the remaining \textsc{writeMax} operations and the \textsc{readMax} operations.
We go through these operations by the order they are invoked 
and place each one after a visible \textsc{writeMax} operation.
Each operation is placed immediately before the next visible \textsc{writeMax} operation in the order, or at the end if no such visible operation exists.

To place the \textsc{readMax} operations, we need to show that they return a value written by some \textsc{writeMax}. Since invalidations occur in descending order, and $S[x]$ is set before invalidating values smaller than $x$, we have the following proposition.

\begin{proposition}
\label{prop:write exists}
    If $I[x] = 1$ for some $x$, then at the same point in the execution, either $I[x+1] = 1$ or $S[x+1] = 1$.
\end{proposition}

Consider a \textsc{readMax} operation $rop$ that returns $x$. To return $x$, $rop$ must either directly read $S[x] = 1$, or read $I[x] = 0$ during an upward scan. By Proposition~\ref{prop:write exists}, the latter case also implies that $S[x]$ contains 1.
This implies that there is a visible \textsc{writeMax}$(x)$ operation that sets this entry, and $rop$ is placed after it.

Consider an invisible \textsc{writeMax}$(x)$ operation $wop$. 
It performs $S[x].\mathsf{write}(1)$ only after $S[x]$ has been previously set.
Consider all visible \textsc{writeMax} operations that successfully set their entries in $S$  prior to the $\mathsf{write}$ invoked by $wop$. We place $wop$ after the visible \textsc{writeMax} operation with the maximal argument among them.
Note that $wop$ is placed after a visible \textsc{writeMax}$(y)$ operation such that $y \geq x$. 

This construction respects the sequential specification of a max-register. Because visible \textsc{writeMax} are ordered by their arguments and no invisible \textsc{writeMax} is placed before a visible \textsc{writeMax} with a smaller or equal argument, every \textsc{readMax} returns the maximum value among all \textsc{writeMax} operations linearized before it.

The next lemma is the key to proving that the ordering preserves the real-time order between non-overlapping operations.
It follows from the fact that \textsc{readMax} operations returning $x$, and \textsc{writeMax}($x$) operations invalidate all $I[y]$ for $y < x$.

\begin{lemma}
\label{lem:read larger}
    Let $op$ be a \textsc{readMax} operation that returns $x$ or a \textsc{writeMax}$(x)$ operation. 
    For any operation $op'$ that follows $op$, where $op'$ is a \textsc{readMax} operation returning $y$ or a visible \textsc{writeMax}$(y)$ operation, it holds that $x \leq y$.
\end{lemma}

Since visible \textsc{writeMax} operations are ordered by their argument, 
Lemma~\ref{lem:read larger} proves that the ordering of the visible \textsc{writeMax} operations and \textsc{readMax} operations respects the real-time order.
The real-time order is naturally preserved between invisible \textsc{writeMax} operations, as well as between a visible \textsc{writeMax} and an invisible \textsc{writeMax}.

It only remains to consider the real-time order between invisible \textsc{writeMax} and \textsc{readMax} operations. 
Consider an invisible \textsc{writeMax} operation $wop$ placed after a visible \textsc{writeMax}$(x)$ operation.
Any \textsc{readMax} operation that precedes $wop$ must return a value smaller than or equal to $x$. 
Moreover, by Proposition~\ref{prop:max val set}, any \textsc{readMax} operation that follows $wop$ must return a value greater than or equal to $x$. 

We apply the same methodology as for the single-writer register (Lemma~\ref{lem:non-leaking-reg}) to prove that the max register is fully-non-leaking.
We construct a minimal execution that is indistinguishable from the original execution to process $p$. 
This minimal execution includes only the \textsc{writeMax} operations necessary to explain the values 
returned by the \textsc{readMax} operations of $p$. 
The proof for the max register is slightly more involved, since \textsc{readMax} operations returning $x$ may view different values from memory, unlike the single-writer register, where a \textsc{read} returning $x$ always views the same unique values.

\begin{restatable}{lemma}{lemmaxregfullyNL}
  \label{lem:maxreg-fully-NL}
Algorithm~\ref{alg:maxregister} is a fully-non-leaking implementation of a bounded max register.  
\end{restatable}

\begin{proof}
  Let $\alpha, \beta$ be executions of respectively the abstract max register and Algorithm~\ref{alg:maxregister}  such that $\alpha$ is the linearization of $\beta$. 
  Let $\alpha' \sim_p \alpha$.
  We show that there exists an execution $\beta' \sim_p \beta$ containing exactly the operations in $\min_p(\alpha')$.
  By Lemma~\ref{lem:fnl reduction}, this implies that the algorithm is fully-non-leaking.
  
  We construct $\beta'$ by iterating over the operations of $p$ in $\min_p(\alpha')$, so that $\beta'$ contains only the operations from $\min_p(\alpha')$ that appear before the operation by $p$.
  Since $\min_p(\alpha') \sim_p \alpha$ and $\alpha$ is the linearization of $\beta$, both $\min_p(\alpha')$ and $\beta$ contain the same operations by process $p$, appearing in the exact same order.

    If the next operation of $p$ is a \textsc{writeMax} operation, we simply append it to $\beta'$. We can do this because $p$ only writes during the operation and never reads from memory. 

    If the next operation of $p$ is a \textsc{readMax} operation returning $x$,
    then since $\min_p(\alpha')$ respects the sequential specification, 
    there is a preceding \textsc{writeMax}$(x)$ in $\min_p(\alpha')$, and no preceding \textsc{writeMax}  in $\min_p(\alpha')$ has argument $y > x$. 
    By construction, the latter also holds in $\beta'$.

    We consider the \textsc{writeMax} operation with the maximum argument (if it exists) in $\beta'$. If the argument for this \textsc{writeMax} is $x$, we simply append the \textsc{readMax} operation by $p$ to $\beta'$.
    Otherwise, there must be some \textsc{writeMax}$(x)$ operation by process $q \neq p$, that precedes the \textsc{readMax} operation in $\min_p(\alpha')$.
    We consider two cases.

    If $p$ returns $x$ without performing an upward scan. That is, $p$ finds $x$ is not yet invalidated ($I[x] = 0)$ and set ($S[x] =1$). We append a complete \textsc{writeMax}$(x)$ operation by process $q$ to $\beta'$, and then append the \textsc{readMax} operation by $p$ to $\beta'$.

    Otherwise, $p$ finds no value that is not invalidated and set while scanning downward the arrays $I$ and $S$, and after seeing some value $y < x$ invalidated, scan the invalidation array upwards.
    Consider a \textsc{readMax} returning $z$ or a \textsc{writeMax}$(z)$ by process $p$ that precedes the \textsc{readMax} operation.
    It must hold that $z \leq y$. This is because after the operation returns $S[z] = 1$, and in the downward scan $p$ reads $0$ from $S[w]$ for any $w > y$.
    Hence, by construction, any \textsc{writeMax} operation in $\beta'$ must have an argument smaller than or equal to $y$. This implies that $I[z] = 0$ and $S[z] = 0$ for any $z > y$.
    
    We let $p$ execute the steps of its \textsc{readMax} operation, pausing it just before it reads $I[y]$ during its downward scan. Then, we append a complete \textsc{writeMax}$(x)$ operation by process $q$ to $\beta'$, and let $p$ complete its remaining \textsc{readMax} steps.

    By construction, $\beta' \sim_p \beta$. In addition, due to the minimality of $\min_p(\alpha')$, any operation not executed by $p$ exists solely to justify a value returned by a \textsc{readMax} operation by $p$. Hence, $\min_p(\alpha')$ does not contain any operation that is not included in $\beta'$. 
\end{proof}

\section{No Wait-Free Unbounded Fully-Non-Leaking Max Register}

The fully-non-leaking implementation of a bounded max register from binary registers 
(Algorithm~\ref{alg:maxregister})
can accommodate an unbounded domain by using infinite arrays. 
Unfortunately, with this change, writing increasing values 
may prevent a \textsc{readMax} from finding a valid marker ($I[x] = 0$), 
and thus, the implementation is only \emph{lock-free}.
On the other hand, if we allow base objects with infinite state, 
there is a wait-free fully-non-leaking unbounded max register 
from unbounded responseless CAS registers (see Appendix~\ref{app:max-reg-cas}).

In this section, we prove that compromising one of these properties is inherent: 
there is no wait-free linearizable fully-non-leaking unbounded max-register, 
from base objects with a finite number of states. 
There is no restriction on the \emph{number} of base objects. 

The proof considers two processes, a writer $w$ and a reader $r$.
For the lower bound, we use the following property of 
fully-non-leaking implementations: 
a process is unaware of the steps of \textsc{readMax} operations performed by other processes.\footnote{This property also holds if the reader never writes. 
Hence, this lower bound also shows that a reader must write in any wait-free implementation of an unbounded max register from finite base objects.}
Hence, assuming an execution where $r$ only performs \textsc{readMax} operations and $w$ performs only \textsc{writeMax} operations, the writer always performs exactly the same steps regardless of the steps of the reader. 
More formally, for any execution $\beta$ in which process $w$ performs only \textsc{writeMax} operations (alongside any number of \textsc{readMax} operations by the other process), 
there is an execution $\beta'$ consisting exclusively of $w$'s $\textsc{writeMax}$ operations such that $\beta \sim_{w} \beta'$.

Clearly, as each base object holds only a finite number of states, and \textsc{writeMax} operations can write arbitrarily large values, the number of objects a writer needs to modify is unbounded. Intuitively, what the lower bound shows is that this also causes the reader to access an unbounded number of objects. Therefore, to guarantee that the reader eventually returns, the reader must notify the writer of its presence, and the writer must record a value that the reader is allowed to return.
However, because we require the implementation to be non-leaking, the writer is unaware of the reader's request. 

The proof considers a restricted set of executions of the following form.
The writer, $w$, performs an infinite sequence of \textsc{writeMax} operations, 
writing $1, 2, 3, \ldots$.
Interleaved between these operations are steps by the reader $r$; 
we will show that these steps belong to a single \textsc{readMax} operation.
Such an execution is specified stating the indexes of the \textsc{writeMax} 
operations after which $r$ takes steps.

Such an execution is determined by a \emph{$k$-sequence}, $k \geq 0$, 
namely, a sequence of $k$ positive integers $\tau = t_1 < t_2 < \dots < t_k$;
for $k=0$, $\tau$ is the empty sequence.
For a $k$-sequence $\tau = t_1 < t_2 < \dots < t_k$, 
$\alpha(\tau)$ is the execution in which the writer $w$ performs an infinite sequence of \textsc{writeMax}$(j)$ operations, 
and the reader $r$ executes $k$ steps interleaved among these writes, as follows:
$r$ takes its $i$-th step, $1 \leq i \leq k$, after the $t_i$-th and before the $(t_i+1)$-th
\textsc{writeMax} operation by $w$. 
We abuse terminology by using ``point $t_i$'' to denote the exact moment in the execution immediately before the reader $r$ takes its $i$-th step.

For $k > 0$, an infinite set of $k$-sequences is \emph{unlimited} 
if their first entries are distinct;
for $k=0$, the singleton set containing the empty sequence is unlimited.

A set $\mathcal{T}_{k+1}$ of $(k+1)$-sequences \emph{extends} a set 
$\mathcal{T}_k$ of $k$-sequences 
if every sequence $\tau \in \mathcal{T}_k$ has exactly one extension 
$\tau' \in \mathcal{T}_{k+1}$, 
and every sequence $\tau' \in \mathcal{T}_{k+1}$ is an extension of 
some sequence $\tau \in \mathcal{T}_k$.
(That is, there is a one to one mapping from $\mathcal{T}_k$ onto $\mathcal{T}_{k+1}$, 
in which each $k$-sequence is mapped to a unique $(k+1)$-sequence that extends it.)

\begin{lemma}
\label{lem:max-reg-finite}
For every $k \geq 0$ there is an unlimited set of $k$-sequences $\mathcal{T}_k$ such that 
$\alpha(\tau_1) \sim_{r} \alpha(\tau_2)$ for every pair $\tau_1, \tau_2 \in \mathcal{T}_k$; 
furthermore, if $k > 0$, then 
$\mathcal{T}_k$ extends $\mathcal{T}_{k-1}$. 
\end{lemma}

\begin{proof}
We prove the lemma by induction on $k$, maintaining the 
additional property that the set $\mathcal{T}_k$ is also \emph{extendable} in the following sense:
there is an infinite set of positive integers $I_{k}$ such that, for every base object $b$, there is a state $s_b$ such that for any sequence $\tau \in \mathcal{T}_k$ and any integer $t \in I_{k}$, if concatenating $t$ to $\tau$, denoted $\tau \cdot t$, yields a valid $(k+1)$-sequence, then in the execution $\alpha(\tau \cdot t)$, object $b$ is in state $s_b$ at point $t$. 
\textbf{Base case:}
$\mathcal{T}_0$, which contains only the empty tuple, 
is unlimited by definition. 
Furthermore, because every base object has only a finite number of states, for any base object $b$, there is at least one state $s_b$ that $b$ is in infinitely often during the execution $\alpha(())$.
Therefore, the empty sequence $()$ has infinitely many distinct $1$-extensions where the object $b$ is in state $s_b$ immediately before $r$'s first step.
This implies that $\mathcal{T}_0$ is extendable.

\textbf{Inductive step:}
    By the inductive assumption, all executions $\alpha(\tau)$ for $\tau \in \mathcal{T}_k$ are indistinguishable to $r$. Therefore, in its $(k+1)$-th step, $r$ accesses the same base object $b$. 

    Because $\mathcal{T}_k$ is extendable, choosing any of the guaranteed extensions for each $\tau \in \mathcal{T}_k$ ensures that $\alpha(\tau_1) \sim_{r} \alpha(\tau_2)$ for all extensions $\tau_1$ and $\tau_2$. 
    We now show how to pick these extensions so that the extendability condition holds for the base object $b$.

    Let $t, t' \in I_{k}$ such that $t < t'$. If we extend any sequence $\tau \in \mathcal{T}_k$ by appending $t$, we know $b$ is in state $s_b$ at point $t$ in $\alpha(\tau \cdot t)$.

    Because the implementation is non-leaking, $b$ is in the same state at point $t'$ in $\alpha(\tau \cdot t \cdot t')$ for all $\tau \in \mathcal{T}_k$. 
    This is, because the $(k+1)$-th step of process $r$ accesses and modifies $b$ in the same way across all $\tau \cdot t$ sequences, and in a non-leaking implementation the writer $w$ performs the same sequence of base operations
    after point $t$ regardless of the reader's steps.
    
    We color a pair $(t, t')$ by the state of $b$ at point $t'$. 
    Since the base object $b$ has a finite number of possible states, the number of colors is finite. 
    By the infinite Ramsey theorem, there exists an infinite subset $I_{k+1} \subseteq I_{k}$ and state $s'_b$ of $b$
    such that for every pair $t, t' \in I_{k+1}$ with $t < t'$ and $\tau \in \mathcal{T}_k$, the object $b$ is guaranteed to be in state $s'_b$ at point $t'$ in $\alpha(\tau \cdot t \cdot t')$.
    
    To construct $\mathcal{T}_{k+1}$, 
    consider a sequence $\tau \in \mathcal{T}_k$.
    There are infinitely many integers in $I_{k+1}$ strictly greater than the $k$-th entry of $\tau$;
    we select such an integer $t \in I_{k+1}$ and append it to $\tau$ to obtain a $(k+1)$-sequence $\tau' = \tau \cdot t$. 
    $\mathcal{T}_{k+1}$ is the set of all these $\tau'$ sequences. 

    Next, we prove extendability for any base object $b' \neq b$. By the inductive hypothesis,
    $b'$ is in state $s_{b'}$ at point $t$ in $\alpha(\tau \cdot t)$. 
    Let $t'\in I_{k+1}$ such that $t < t'$.
    Since $I_{k+1} \subseteq I_k$, $t'\in I_k$.
    Because the implementation is non-leaking, in the execution $\alpha(\tau \cdot t \cdot t')$, the object $b'$ must be in state $s_{b'}$ at point $t'$. 
    This holds because the writer $w$ performs the same primitive operations regardless of the reader's steps, and $r$'s $(k+1)$-th step only accesses $b$, leaving the state of $b'$ completely unchanged.

    For $k=0$, the set $\mathcal{T}_1$ is formed by an infinite number of distinct $1$-extensions of the empty tuple. Because these are distinct $1$-sequences, their first entries are strictly distinct, making $\mathcal{T}_1$ unlimited by definition.
    For $k\geq 1$, since $\mathcal{T}_{k+1}$ extends $\mathcal{T}_k$ and $\mathcal{T}_k$ is unlimited, 
    we have that $\mathcal{T}_{k+1}$ is also unlimited.

    Because $r$'s local view is identical prior to its $(k+1)$-th step, and it observes the identical state $s_b$ during this step across all executions, $\alpha(\tau_1) \sim_{r} \alpha(\tau_2)$ for all $\tau_1, \tau_2 \in \mathcal{T}_{k+1}$.
\end{proof}

\begin{theorem} 
\label{thm:lower bound}
    There is no wait-free implementation of an unbounded fully-non-leaking max register for two processes using only base objects that each have a finite number of states.
\end{theorem}

\begin{proof}
   We construct an infinite sequence of increasingly longer executions, in each of which the reader $r$ takes infinitely many steps but never returns from its \textsc{readMax} operation.

   By Lemma~\ref{lem:max-reg-finite}, for every $k \ge 1$, there is an unlimited set of $k$-sequences $\mathcal{T}_k$ such that $\mathcal{T}_k$ extends $\mathcal{T}_{k-1}$ and $\alpha(\tau_1) \sim_{r} \alpha(\tau_2)$ for all $\tau_1, \tau_2 \in \mathcal{T}_k$.

    Let $k\geq 1$ and $\tau \in \mathcal{T}_k$.
    Suppose, by way of contradiction, 
    that the reader returns a value
    $v$ at some point in $\alpha(\tau)$.
    Then it returns the same value $v$
    at some point in each execution $\alpha(\tau')$
    for each $\tau' \in \mathcal{T}_k$,
    as it cannot distinguish these executions.

    Because $\mathcal{T}_k$ is unlimited,
    there is a sequence $\tau' \in \mathcal{T}_k$ whose first entry is $> v$. 
    In the execution $\alpha(\tau')$, the writer $w$ completes a \textsc{writeMax}($v+1$) 
    operation before the reader $r$ takes its first step. 
    Thus, the \textsc{readMax} operation returning $v$ when 
    a \textsc{writeMax}($v+1$) precedes it 
    violates the linearizability of a max register.
Thus, we get increasingly longer executions, 
    in which the reader taking an unbounded number of steps without returning,
    contradicting wait-freedom.    
\end{proof}

\section{Argument-Non-Leaking Wait-Free Stack and Lock-Free Queue}
\label{sec:stack_queue}

We next turn to argument-non-leaking implementations.
We adapt existing implementations of a stack~\cite{AfekGM2007} 
and a queue~\cite{li2001} by applying a common technique from the world of secret-sharing: randomly splitting a value into two parts,
each of which gives no information about the value in isolation.
To protect the arguments of the stack and the queue from leakage,
we write each part into a different location in memory,
so that a process does not learn the element unless it reads both locations.
This makes correctness somewhat more tricky, as now an element is not stored in one place and cannot be accessed atomically.

\subsection{Wait-Free Stack}

Our stack is based on the wait-free stack of Afek, Gafni and Morrison~\cite{AfekGM2007}, reproduced on the left side of Algorithm~\ref{alg:side-by-side stack}.
(In the same paper the authors also present a swap-based stack, which is already argument-non-leaking.)
This implementation leaks operations and arguments.
For example, suppose that when process $p$ invokes \textsc{push} for the first time and calls $i \leftarrow range.\mathsf{fetch\&add}$,
the value fetched is $i = 2$.
Then process $p$ learns that some other process invoked \textsc{push} 
before it, even though the abstract semantics of the stack does not reveal this information.
This is \emph{operation leaking} ($\APop$).
But the implementation also leaks \emph{values}:
suppose that process $p$'s first access to the stack
is to invoke \textsc{pop},
and it reads a non-$\bot$
value $y$ from some entry $items[i]$,
but then fails the $\mathsf{test\&set}$ on $T[i]$,
eventually returning a different value $z \neq y$
from some other entry $items[j]$ where $j < i$.
Then process $p$ learns that $y$ was pushed onto the stack by some other process,
even though its interaction with the stack should not reveal this information according to the abstract semantics:
it should only have learned that the value $z$ was pushed by some other process, and should not have learned about $y$.
To address this, 
in our modified stack,
instead of writing input $x$ in a single register $items[i]$, a \textsc{push} operation (Algorithm~\ref{alg:side-by-side stack} (right)) randomly splits $x$ into two parts $x_1$ and $x_2$ whose sum is $x$.
These parts are written into two distinct registers $R[i]$ and $S[i]$.
Crucially, $S[i]$ is written first.

A \textsc{pop} operation $op$ trying to obtain the $i$-th value pushed into the stack first checks that the corresponding \textsc{push} operation has terminated by reading $R[i]$. 
If $R[i] \neq \bot$, the process tries to claim slot $i$ by performing a $\mathsf{test\&set}$ on $T[i]$. Only if this slot is successfully claimed does the \textsc{pop} operation read $S[i]$, retrieve $x$, and return it.
If $x$ has already been returned by another \textsc{pop}, then the  $\mathsf{test\&set}$ on $T[i]$ fails, and $op$ cannot be sure what is the input of the $i$-th push operation.
In this case the \textsc{pop} operation gives up on obtaining the $i$-th value and attempts to obtain the $(i-1)$-th value instead.
This continues until either a value is successfully obtained,
or the operation reaches $i = 0$ and returns the special value $\mathrm{EMPTY}$.

We remark that although an operation does not learn the arguments or return values of other operations, it can sometimes infer their existence:
for example, a \textsc{pop} that reads the value $t$ from $range$ learns that at least $t-1$ \textsc{push} operations have started. 
Similarly, a failed $\mathsf{test\&set}$ on an entry in $T$ leaks the existence of a \textsc{pop} operation.

\begin{algorithm}[bt]
  \caption{Afek, Gafni and Morrison's wait-free stack (left) and its non-leaking variant (right)}
  \label{alg:side-by-side stack}
  \begin{algorithmic}[1]
  \small\Statex
    \begin{minipage}[t]{0.49\textwidth}
    \State\textbf{Shared objects}
    \State \textit{range}: fetch\&add object, initially $1$
    \State \textit{items}: array $[1,\ldots]$ of registers; initially $\bot$
    \State
    \State $T$: array $[1,\ldots]$ of test\&set objects; 
    \Statex \quad\quad\quad initially $\bot$

    \Function{push}{$x$}
      \State $i \gets range.\mathsf{fetch\&add}(1)$       
     \State
     \State $\textit{items}[i].\mathsf{write}(x)$    
    \EndFunction

    \State
    \Function{pop}{$~$}
      \State $t \gets range.\mathsf{fetch\&add}(0) - 1$
      \For{ $i$ from $t$ downto $1$}
        \State $y \gets \textit{items}[i].\mathsf{read}()$
        \If{$y \neq \bot$}
          \If{$T[i].\mathsf{test\&set}$()} 
          \State
          \State \Return{$y$}
          \EndIf
        \EndIf
      \EndFor
      \State \Return{EMPTY}
      \EndFunction
    \end{minipage}
    \begin{minipage}[t]{0.46\textwidth}
    \setcounter{ALG@line}{0}
    \State\textbf{Shared objects}
    \State $range$: fetch\&add object, initially $1$
    \State $R$: array $[1,\ldots]$ of registers; initially $\bot$
    \State $S$: array $[1,\ldots]$ of registers; initially $\bot$
    \State $T$: array $[1,\ldots]$ of test\&set objects;
    \Statex \quad\quad\quad initially $\bot$

    \Function{push}{$x$}
      \State $i \gets range.\mathsf{fetch\&add}(1)$ 
      \State $N \gets$ random number 
      \State $S[i].\mathsf{write}(N)$
      \State $R[i].\mathsf{write}(x+ N)$ 
    \EndFunction

    \Function{pop}{$~$}
      \State $t \gets range.\mathsf{fetch\&add}(0) - 1$
      \For{ $i$ from $t$ downto $1$}
        \State $y \gets R[i].\mathsf{read}()$
        \If{$y \neq \bot$}
            \If{$T[i].\mathsf{test\&set}()$}
                \State $N = S[i].\mathsf{read}()$
                \State \Return{$y-N$}
            \EndIf
        \EndIf
      \EndFor
      \State \Return{EMPTY}
      \EndFunction
      \end{minipage}
  \end{algorithmic}
\end{algorithm}

Our proof for the linearizability of Algorithm~\ref{alg:side-by-side stack} follows~\cite{AfekGM2007}.
In particular, we reuse their method of constructing a sequential 
history \emph{after} the complete concurrent history is known, 
rather than assigning explicit linearization points to operations.

Given a finite execution $\beta$, 
they construct a \emph{sequential} history $Seq(\beta)$ by going 
over the events of $\beta$, so that $Seq(\beta)$:
(i) satisfies the stack specification,
(ii) respects the real-time order of non-overlapping operations in $\beta$, and
(iii) contains all operations completed in $\beta$.

We reuse this construction verbatim, 
with minor adaptations to account for the different per-slot claiming mechanism and the encoding of values used in our algorithm.

Each \textsc{push} operation obtains a unique index $i$ by executing $range.\mathsf{fetch\&add}(1)$.
The pushed value becomes visible to \textsc{pop} operations when a non-\texttt{NULL} value is written to $R[i]$.
As in AGM, indices totally order \textsc{push} operations and represent abstract stack positions.

 In both AGM's and our stack, a \textsc{pop} operation claims slot $i$ by winning a
$\mathsf{test\&set}$ on that slot.
Thus, Algorithm~\ref{alg:side-by-side stack} follows the same structure of the proof: 
in both algorithms, each slot can be claimed by at most one \textsc{pop}.

\begin{lemma}
\label{lem:unique-claim}
For every index $i$, 
at most one \textsc{pop} operation in $\beta$ returns a value from slot $i$.
\end{lemma}

Algorithm~\ref{alg:side-by-side stack} (right) encodes each pushed value $x$ by writing 
a random value $N$ to $S[i]$ and writing $x+N$ to $R[i]$.
A \textsc{pop} that successfully claims slot $i$ retrieves 
$S[i]$ and $R[i]$, and returns their difference. 

\begin{lemma}
\label{lem:correct-value}
If a \textsc{pop} operation returns from slot $i$, then it returns exactly the
value written by the unique \textsc{push} operation that published slot $i$.
\end{lemma}

We now construct the sequential history $Seq(\beta)$ exactly as in AGM.
Briefly, $Seq(\beta)$ is built by scanning the events of $\beta$ 
and maintaining an auxiliary abstract stack.
When a \textsc{push} publishes slot $i$, the corresponding \textsc{push}
operation is appended to $Seq(\beta)$.
When a \textsc{pop} successfully claims slot $i$, 
the corresponding \textsc{pop} operation is appended to $Seq(\beta)$ 
and matched with the most recent unmatched \textsc{push} at index $i$.
A \textsc{pop} that finds no claimable slot is appended as an empty pop.

By Lemmas~\ref{lem:unique-claim} and~\ref{lem:correct-value}, 
this construction is well-defined and mirrors the AGM construction.
As shown by AGM, the resulting sequential history $Seq(\beta)$ is 
a legal stack history: 
each \textsc{pop} returns the value of the most recent unmatched
\textsc{push}, or EMPTY if none exists.
Moreover, the construction respects the real-time order of non-overlapping
operations, since a \textsc{pop} can only be matched with a \textsc{push} whose
publication precedes the claim.

Since the construction of $Seq(\beta)$ and the associated correctness 
arguments are identical to those of AGM (using the two lemmas above), 
it follows that Algorithm~\ref{alg:side-by-side stack} (right) implements a linearizable stack.

Finally, we argue that Algorithm~\ref{alg:side-by-side stack} 
does not leak the arguments or return values of operations.

\begin{lemma}
    Algorithm~\ref{alg:side-by-side stack} (right) is argument-non-leaking.
    \label{lemma:stack_nonleaking}
\end{lemma}

\begin{proof}
    Let $X$ be the domain of values that can be pushed into the stack.
    We begin by observing that for any execution $\beta$,
    process $p$ and set of values $X' \subsetneq X$,
    if $p$ does not itself push or pop a value from $X'$ in $\beta$,
    then there is an indistinguishable execution $\beta' \sim_p \beta$
    where \emph{no process} pushes or pops any value from $X'$.
    This is because for any value $x \in X'$ that is not returned from a \textsc{pop}
    operation by process $p$,
    throughout $\beta$,
    process $p$ never reads $R[i], S[i]$ such that $R[i] - S[i] = x$
    for any $i$
    (process $p$ only ever performs a successful $\mathsf{test\&set}$ on $T[i]$ inside a \textsc{pop} operation,
    immediately prior to returning $R[i] - S[i]$).
    Thus, we can construct $\beta'$ from $\beta$
    as follows:
    first, choose some value $y \in X \setminus X'$
    arbitrarily (such a value exists since $X' \subsetneq X$).
    In any \textsc{push}$(x)$ operation where $x \in X'$,
    we replace $x$ by $y$,
    and inside the operation we replace the value of $N$
    by $N' = x - y + N$,
    so that the value written into $R[i]$ is the same
    in $\beta, \beta'$
    (as $y + N' = x + N$).
    This implies that whenever process $p$ reads $R[i]$ it will
    read the same values in $\beta, \beta'$.
    As we said above,
    process $p$ never reads $S[i]$ unless it is about to return
    $R[i] - S[i]$ from a \textsc{pop},
    and therefore it does not observe the different value $N'$
    that is written into $S[i]$ (as we assumed that process $p$ does not 
    return $y \in X'$ from a \textsc{pop}).
    Finally, for \textsc{pop} operations we do not need to make any further
    changes,
    as any non-$\mathrm{EMPTY}$ value returned by \textsc{pop} was previously
    \textsc{push}ed,
    and in $\beta'$ no value from $X'$ is \textsc{push}ed.

    We now prove that the algorithm is argument-non-leaking.
    Since \textsc{push} operations do not have a return value and \textsc{pop} operations do not take an argument, it is convenient to denote the
    atomic propositions for these operations by $op(m, V)$ 
     where $m \in \set{\text{\textsc{push}, \textsc{pop}}}$,
    and $V \subsetneq P \times X$.
    The proposition $op(\text{\textsc{push}}, V)$
    corresponds to a \textsc{push} operation performed by a process $q$
    of a value $x$ such that $(q,x) \in V$,
    and the proposition $op(\text{\textsc{pop}}, V)$
       corresponds to a \textsc{pop} operation performed by a process $q$
    that returns a value $x$ such that $(q,x) \in V$.
    
    Let $\beta$ be a concrete execution,
    let $\alpha = r(\beta)$ where $r$ is the linearization mapping defined
    above (it is not important for our purposes here;
    any legal mapping would do),
    and let $V \subsetneq P \times X$ be a set of 
    pairs of process ID and element
    such that
    $\alpha \not \models \K[p][op(m, V)]$ for $m \in \set{\text{\textsc{push}, \textsc{pop}}}$.
    Then there is an indistinguishable abstract execution 
    $\alpha' \sim_p \alpha$
    where $\alpha' \not \models op(m, V)$,
    that is,
    in $\alpha'$,
    no operation $op(m, q, x)$ such that $(q, x) \in V$
    appears.
    In particular this implies that in $\alpha'$, process $p$
    itself does not invoke $op(m,q,v)$
    for any value $x$ such that $(p, x) \in V$,
    and since $\alpha' \sim_p \alpha$,
    the same holds in $\alpha$.
    Finally, since $\alpha = r(\beta)$,
    the same holds in $\beta$.
    By our claim above, there is an indistinguishable
    execution $\beta' \sim_p \beta$
    where for any $(q,x) \in V$,
    process $q$ does not invoke \textsc{push}(x)
    or return the value $x$ from \textsc{pop}.
    This means $\beta' \not \models op(m, V)$,
    and since $\beta' \sim_p \beta$,
    we have $\beta \not \models \K[p][op(m, V)]$,
    as desired.  
\end{proof}

\subsection{Lock-free Queue}
\label{sec:queue}

Li's queue~\cite{li2001} is  similar to the AGM stack; 
the only difference is that enqueues add elements to the tail, 
while dequeues remove elements from the head. 
Like the AGM stack, it leaks \textsc{enqueue} operations.
However, we can make it non-leaking with the same technique of randomly splitting each enqueued value $x$ into two parts whose sum is $x$.

\begin{algorithm}[bt]
  \caption{Non-leaking lock-free queue, adapted from Li~\cite{li2001}}
  \label{alg:queue}
  \begin{algorithmic}[1]
  \small
    \State\textbf{Shared objects}
    \State\hspace{\algorithmicindent}$range$: fetch\&add object, initialized to $1$
    \State\hspace{\algorithmicindent}$R$: array $[1,\ldots]$ of registers (each initially NULL)
    \State\hspace{\algorithmicindent}$S$: array $[1,\ldots]$ of registers (each initially NULL)
    \State\hspace{\algorithmicindent}$T$: array $[1,\ldots]$ of test\&set objects (each initially NULL)

    \Function{enqueue}{$x$}
      \State $t \gets range.\mathsf{fetch\&add}(1)$
      \State $N \gets$ random number 
      \State $S[t].\mathsf{write}(N)$
      \State $R[t].\mathsf{write}(x + N)$ 
    \EndFunction

    \Function{dequeue}{$~$}
      \State $t \gets 0$; $emptied \gets 0$
      \Repeat
        \State $t_{old} \gets t$
        \State $t \gets range.\mathsf{fetch\&add}(0) - 1$
        \State ${emptied}_{old} \gets {emptied}$; ${emptied} \gets 0$
        \For{ $i$ from $1$ upto $t$}
          \State $y \gets R[i].\mathsf{read}()$
          \If{$y \neq \mathrm{NULL}$}
            \If{$T[i].\mathsf{test\&set}()$}
                \State $N = S[i].\mathsf{read}()$
                \State \Return{$y-N$}
            \EndIf
            \State $emptied \gets emptied + 1$
          \EndIf
        \EndFor
      \Until{$t = t_{old}$ and $emptied = emptied_{old}$}
      \State \Return{EMPTY}
      \EndFunction
  \end{algorithmic}
\end{algorithm}

Our non-leaking queue is presented in Algorithm~\ref{alg:queue}.
The \textsc{dequeue} operation reads the current value of \textit{range} and iterates over the elements from the first one down to the tail. Unlike a \textsc{pop}, it cannot simply return EMPTY if it does not find an element to return. 
It can only do so if it observes that no new elements have been enqueued since its last traversal. Otherwise, it rereads the current value of the tail and restarts the traversal.

As in the AGM stack, both enqueues and dequeues claim slots in the infinite array. 
By using arguments similar to the proof of the stack algorithm
we can follow the proof of the original algorithm and show that 
it is linearizable and lock-free.
The argument-non-leaking property is also similar:
for each value $x$ that is not enqueued or dequeued by a process $p$,
we can construct an execution that is indistinguishable to $p$,
where no process enqueues or dequeues $x$.
This allows us to argue that for any concrete execution $\beta$ where the abstract specification
does not allow $p$ to learn some proposition $op(m, V)$,
there is an indistinguishable concrete execution $\beta' \sim_p \beta$ where $op(m, V)$ in fact
does not hold, and therefore $p$ does not learn $op(m, V)$ in $\beta$.

\section{Input-Non-Leaking Wait-Free Approximate Agreement}
\label{sec:approx_agreement}

Finally, we extend our framework to capture information leaking in the context of decision tasks, using approximate agreement as an example.
Algorithms for approximate agreement are often based on  successive halving phases, in each of which, 
processes write their current proposed values to the shared memory, collect all the values, and take their mean or median~\cite{AttiyaLS1994}.
This structure clearly allows processes to know the input values of 
other processes, even when they end up deciding on their own input. 
We show, however, that a slight adaptation of Schenk's approximate agreement algorithm~\cite{Schenk95} does not leak inputs. 

\subsection{Defining Input-Non-Leaking Approximate Agreement}

The framework of Section~\ref{sec:definitions} can be 
used to define non-leaking \emph{decision tasks}.
For simplicity, we concentrate on colorless tasks, 
and approximate agreement, in particular. 

A \emph{colorless} task  $T$ is a triplet $(\Delta,\m{I},\m{O})$, where  $\Delta$ maps each set of inputs in $\mathcal I$ to a set of allowed outputs in $2^{\mathcal O}$, subject to for any $\sigma \subset \sigma' \in \m{I}$, $\Delta(\sigma) \subseteq \Delta(\sigma')$.
Equivalently, $T$ may be specified as an object $O_T = (Q,\Sigma,s_0,\delta)$ with a single method (hence omitted). Intuitively, each state $q \in Q$ records inputs received and outputs returned, so that future return actions  output values allowed by $\Delta$, and consistent with past outputs. Formally, $\Sigma = Call \cup Ret$, $Q \subseteq \m{I} \times \m{O}$ with $s_0 = (\emptyset,\emptyset)$. For a call action $Call(v,\_)$, $(q,Call(v,\_),q') \in \delta$ if and only if $q = (I,O)$, $q' = (I',O')$ with $I' = I \cup \{v\}$ and $O = O'$. A value $v \in \m{O}$ may be returned given state $q = (I,O)$  if this is allowed by the task relation $\Delta$. That is,   $(q,Ret(v,\_),q') \in \delta$ if and only if $q = (I,O)$, $q' = (I',O')$ with $I' = I$, $O = O' \cup \{v\}$ and $O' \in \Delta(I')$. 
Note that $O_T = (Q,\Sigma,s_0,\delta)$ is a deterministic LTS.

Therefore, an algorithm $C$ solves $T$ if $C$ refines $O_T$. Refinements suitable for objects defined from tasks have been identified in~\cite{CastanedaRR18, Neiger94}. 

An algorithm $L$ solving a colorless task  $T = (\Delta, \m{I}, \m{O})$, for a set process $P$, 
is \emph{input-non-leaking} if $L$ is $\APin$-non-leaking, where $\APin$ is the following set of atomic propositions:
  \begin{equation*}
  \APin = \set{ 
    \medspace
    \exists p \in P \medspace
    \exists v \in \m{O} 
    \medspace.\medspace
    op(p, u, v) 
    \quad : \quad
    u \in \m{I}
    \medspace
  }   ,
\end{equation*}
and $op(p,u,v)$ holds in an execution $\gamma$ if process $p$ has input $u$ and outputs $v$.

In \emph{approximate agreement}, each process $p_i$ has an input $x_i$ in some interval $X$ and is required to output a value $y_i$ with 
$ \min(X) \leq y_i \leq \max(X)$
and such that all outputs are within distance $\epsilon$ of one another, where $\epsilon > 0$.

We can fix the input range to be  $[0,1]$ and $\epsilon = 1/N$ for some integer $N$.  
A fixed input range $[a,b]$ and $\epsilon$ can be handled 
by scaling the inputs (i.e, $\bar{x} = (x-a)/(b-a)$) so that they are in the interval $[0,1]$, setting $\bar{\epsilon} = \lceil\epsilon/(b-a)\rceil$, and returning a scaled output (i.e., $y = \bar{y}(b-a)  +a$).
An approximate agreement algorithm is input-non-leaking 
if it is $\APin$-non-leaking, 
for $\APin = \set{ 
    \medspace
    \exists p \in P \medspace
    \exists y \in [0,1] 
    \medspace.\medspace
    op(p, x, y) 
    \medspace : \medspace
    x \in [0,1]
    \medspace
  }$.

\subsection{Schenk's Algorithm (Slightly Modified)}

At the heart of Schenk's approximate agreement algorithm for known initial range~\cite{Schenk95} are \emph{halving objects} (Algorithm~\ref{alg:halving}).  
Given inputs drawn from a known interval $I = [a,b]$, 
a halving object $H_I$ with parameter $I$ allows processes to get outputs that are in the range of the inputs, 
and in one half of $I$. 
In other words, it is an implementation of approximate agreement with input range $I$ and $\epsilon = |I|/2$, with the additional \emph{partition property} that outputs are either all in the lower half of $I$, 
or all in the upper half of $I$. 
The algorithm uses two binary registers $R_0$ and $R_1$, which signal the existence of an input in the lower and  upper half of $I$ respectively. 
A process $p$ sets one of the registers to \emph{true}, according to which half of $I$ contains its input $x$. 
It then reads the other register, and returns $x$ if it is not set (indicating the absence of input in the half not containing its input), and the midpoint of $I$  otherwise.

  \begin{algorithm}[bt]
  \caption{Halving Object $H_{[a,b]}$} 
  \label{alg:halving}
  \begin{algorithmic}[1]
    \small
    \State\textbf{Shared variables}
    $R_0,R_1$ : binary multi-writer registers, initially $\emph{false}$. 

    \Function{halve}{$x$}\Comment{$x \in [a,b]$}
      \State \textbf{if} $x \leq (a+b)/2$ \textbf{then} $k \gets 0$ \textbf{else} $k \gets 1$
      \State $R_k.\mathsf{write}(\emph{true})$
      \State\label{l:halve:new} {\color{blue} \textbf{if} $x =  (a+b)/2$ \textbf{then} \Return{$x$}} 
      \State \textbf{if} $R_{1-k}.\mathsf{read}()$ \textbf{then} \Return{$(a+b)/2$} \textbf{else} \Return{$x$}
    \EndFunction
  \end{algorithmic}
\end{algorithm}

Line~\ref{l:halve:new} (marked in blue), \emph{was added to the original code, and it does not} affect correctness, as any $\mathsf{halve}$ operation with input $x = (a+b)/2$ returns  $x$ after writing to the register corresponding to the lower half of $[a,b]$. 
However, it prevents the invoking process from learning the existence of other operations.

\begin{lemma}[\cite{Schenk95}]
  \label{lem:halving_correct}
  Algorithm~\ref{alg:halving} with parameter $I = [a,b]$ implements approximate agreement with  input range $[a,b]$ and $\epsilon = (b-a)/2$ and ensures the \emph{partition property}: In  every execution, either all outputs are contained in $[a,(a+b)/2]$ or they are all contained  in $[(a+b)/2,b]$. 
\end{lemma}

Since writing to register $R_k$, $k \in \{0,1\}$, 
only reveals which half of $I$ $p$'s input lies in,  Algorithm~\ref{alg:halving} is input-non-leaking.

\begin{lemma}
  \label{lem:halving_nl}
Let $\beta$ be an execution of   Algorithm~\ref{alg:halving} with parameter $I = [a,b]$ in which process $p$'s input is $x$. For every $x' \in [a,b]$, there is $\beta'\sim_p\beta$ in which no process has input $x'$.
\end{lemma}

\begin{proof}
  Suppose, without loss of generality, that $x$ lies in the lower half of $I$. If $p$ reads false from $R_1$, $\beta$ is indistinguishable for $p$ from a solo execution. Otherwise, as $R_1$ is set to \emph{true} by any process with input $> (a+b)/2$, the input of every process but $p$ can be replaced by any value $x'' \neq x' > (a+b)/2$ without $p$ noticing it. 
\end{proof}

Schenk's algorithm (Algorithm~\ref{alg:sch}) proceeds in \emph{phases}.  
Each process $p$ maintains a preference $y_p$, initially its input. 
Each phase aims at halving the range of processes' preferences, using halving objects, until they are all contained in an interval of size $\leq \epsilon$. 
For clarity, assume $\epsilon = 1/2^N$ where $N$ is an integer. 
In each phase $\ell, 0 \leq \ell \leq N-1$, the input range $[0,1]$ is subdivided into intervals of the form $S_{k,\ell} = [k/2^\ell,(k+1)/2^\ell]$, 
for an integer $k$, $0 \leq k < 2^\ell$,
each with a designated halving object $H_{S_{k,\ell}}$. 
The following is maintained:  
All preferences at the beginning of phase $\ell$ are contained in a common sub-interval $S_{c,\ell}$. 

\begin{algorithm}[tbp]
  \caption{Schenk's approximate agreement, $\epsilon = 1/2^N$ for some integer $N$} 
  \label{alg:sch}
  \begin{algorithmic}[1]
    \small
    \State\textbf{Shared variables}
    For each  interval $S_{k,\ell} = [k/2^\ell,(k+1)/2^\ell]$: $0 \leq \ell \leq N, 0 \leq k < 2^\ell$:  halving object $H_{S_{k,\ell}}$ 
    
    \Function{aa$_\epsilon$}{$x$}\Comment{$x \in [0,1]$}
    \State $y \gets x$  \label{l:sch:init}
    \For{$\ell = 0,\ldots, N-1$} \label{l:sch:for}
    \State \textbf{let} $k = \min\{k' : x \in [k'/2^\ell,(k'+1)/2^\ell]\}$;  $d \gets H_{[k/2^\ell,(k+1)/2^\ell]}.\mathsf{halve}(y)$ \label{l:sch:interval_halve}
    \If{$d \neq 0$ and $d = y = (k+1)/2^\ell$ and $k+1 < 2^\ell$}
    $y \gets H_{\range{(k+1)}{(k+2)}{\ell}}.\textsf{halve}(y)$ \label{l:sch:again_halve}
    \Else{} $y \gets d$
    \EndIf
    \EndFor
    \State \Return{$y$}
    \EndFunction
  \end{algorithmic}
  \end{algorithm}

Reducing the range of preferences and maintaining the invariant is easy if every process is able to identify this common  sub-interval $S_{c,\ell}$. 
By invoking $H_{S_{c,\ell}}.\mathsf{halve}(y_p)$, each  process $p$ obtains  a new preference.
By the partition property, they  all lie in the same sub-interval $
S_{2c,\ell+1}$ or $S_{2c+1,\ell+1}$. 

For process $p$, if $y_p$ is strictly contained in $S_{k,\ell}$,  $c$ can be directly inferred  from $y_p$. The main difficulty is when the preference  $y_p$ of a process $p$  has the form $k/2^\ell$, which is in at the boundary of two sub-intervals $S_{k-1,\ell}$ and $S_{k,\ell}$. In that case, $p$ invokes both  halving objects $H_{S_{k-1,\ell}}$ and $H_{S_{k,\ell}}$, in that order.  It is observed in \cite{Schenk95} that only one of the two objects associated with these intervals may return a value $\neq k/2^\ell$. Indeed, if at the beginning of the phase, all preferences are contained in one of the two intervals, for example $S_{k,\ell}$, then every invocation of $H_{S_{k-1,\ell}}$ has the same input $k/2^\ell$ and thus returns this value.

\subsection{Proof of Correctness and Non-Leakage}

Let $\beta$ be an execution of Algorithm~\ref{alg:sch},
  and let $x_{min}$ and $x_{max}$ be the smallest and largest inputs in $\beta$. 
For every iteration $\ell, 0 \leq \ell \leq N-1$ of the for loop, and each process $p$, 
  $y_p^\ell$ is the preference of $p$ at the beginning of the iteration, 
  and $y_p^N$ is the value returned by process $p$. 
Correctness follows from Lemma 1 of~\cite{Schenk95}. 
The  proof of this lemma includes the next claim, 
which we use to prove that the algorithm is input-non-leaking.

\begin{lemma}[Lemma 1 in \cite{Schenk95}]
  \label{lem:sch:correct}
  For all $\ell, 0 \leq \ell \leq N$, there is an integer $c_\ell, 0\leq c_\ell < 2^\ell$ such that  for all processes $p$,  $y_p^\ell \in \range{c_\ell}{(c_\ell+1)}{\ell} \cap [x_{min},x_{max}]$.
\end{lemma}

\begin{lemma}[Claim in the proof of Lemma 1 in \cite{Schenk95}]
  \label{lem:sch:claim}
  For all $\ell, 0 \leq \ell < N$, there is an integer $c_\ell, 0 \leq c_\ell < 2^\ell$, such that  for all processes $p$,  $y_p^{\ell+1}$ is the value returned to $p$ by applying $\mathsf{halve}(y_p^\ell)$ to the object $H_{\range{c_\ell}{(c_\ell+1)}{\ell}}$. 
\end{lemma}

\begin{lemma}
  \label{lem:sch:ind}
  Let $x$ be the input of $p$ in $\beta$. For every $x' \neq x$, there is an execution $\beta' \sim_p \beta$ in which no process has input $x'$.  
\end{lemma}

\begin{proof}
  For each $\ell, 0 \leq \ell < N$,  let $S^\ell$ be the interval ${\range{c_\ell}{(c_\ell+1)}{\ell}}$ defined in Lemma~\ref{lem:sch:claim}. Denote by $S^N$  the half of $S^{N-1}$ to which all final preferences belong to (by the partition property of $H_{S^{N-1}}$, all $y_p^N$ lie in the same half of $S^{N-1}$).

  As any process $p$ invokes $\mathsf{halve}$ on object $H_{S^\ell}$ only if $y_p^\ell \in S^\ell$, $S^0 \supset \ldots \supset S^{N-1} \supset S^N$. 
  
  For the construction of $\beta'$, we use $N+1$ processes $q_0,\ldots, q_N$  with  input $x_0, \ldots, x_N$ respectively. For each $\ell, 0 \leq \ell < N$, we choose  $x_\ell \neq x'$ in the interior of $S^\ell \setminus S^{\ell+1}$, and $x_N \neq x'$  in the interior of $S^N$. Starting from the initial configuration, we let  $q_N$ run until it returns, and then run each process $q_i, 0, \leq i < N$ until it is about to apply $\mathsf{halve}$ on $H_{S^i}$. Denote by $\beta'_0$ this execution. Note that in its solo run, $q_N$  applies  $\mathsf{halve}(x_N)$ to each object $H_{S^0}, \ldots, H_{S^{N-1}}$ (as it is solo), and does not access any other halving objects (as $x_N$ is not at the boundary of any sub-interval $S_k^\ell$). Process $q_\ell$, where $\ell<N$, behaves similarly as $q_N$ in any iteration $\ell' <\ell$, getting back $x_\ell$ from its call to $H_{S^{\ell'}}.\mathsf{halve}()$.  Indeed, for every $\ell' < \ell$, $x_\ell$ and  $x_N$ lie in the same half of $S^{\ell'}$. $q_\ell$ thus writes to the same register of $H_{S^{\ell'}}$ as $q_N$, and, as $q_N$, see the other register in its initial state. This is because  every process $q_j$ with input in the other half of $S^{\ell'}$ stops running before accessing $H^{S^{\ell'}}$.

  We then extend $\beta'_0$ by adding  steps of $p$ and possibly one more step by $q_i, 0 \leq i < N$. Assume inductively that we have constructed prefix $\beta'_{\ell-1}$ and that it is indistinguishable for $p$ from $\beta$ up to the beginning of iteration $\ell$ of the for loop. 

  If $p$ invokes $H_S.\mathsf{halve}(y)$ with $S \neq S^\ell$ in its $\ell$th iteration of the for loop in $\beta$, the invocation returns $y$. 
  Otherwise, $p$ would not pass the test of line~\ref{l:sch:again_halve}, and would update its preference  $y$ with the  value returned. 
  By Lemma~\ref{lem:sch:claim}, this would imply that $S = S^\ell$. Since no process $q_i, 0\leq i \leq N$ modifies  $H_S$, $p$'s invocation  $H_S.\mathsf{halve}(y)$  also returns $y$ in $\beta'$. 
  
  By Lemma~\ref{lem:sch:claim}, every value returned by an invocation of $H_{S^\ell}.\mathsf{halve}()$ lies in the same half (which is $S^{\ell+1}$) of $S^\ell$, and this half contains  $x_N,\ldots,x_{\ell+1}$. Suppose without loss of generality that $S^{\ell+1}$ is the upper half of $S^\ell$. In the following, $R_0$ and $R_1$ refer to the registers in the implementation of $H_{S^\ell}$. Before $p$, only $q_N, \ldots, q_{\ell+1}$ call $H_{S^\ell}.\mathsf{halve()}$, each call with an input in $S^{\ell+1}$. Each of this call therefore writes to  $R_1$ and does not modify $R_0$. Hence $R_1 = true$ and $R_0 = false$ at the end of $\beta'_{\ell-1}$. Let $m$ be the midpoint of $S^\ell$. We consider three cases, depending on which half of $S^\ell$ $p$'s preference is.
    \begin{itemize}
    \item $y_p^\ell < m$. $p$'s invocation of $H_{S^\ell}.\mathsf{halve}$ in $\beta$  returns a value $\geq m$ (actually $m$). Hence, after  writing to $R_0$, $p$ reads $R_1 = \emph{true}$ in $\beta$. We let $p$ perform iteration $\ell$ after $\beta'_{\ell-1}$.  As $R_1 =  \emph{true}$ at the end of $\beta'_{\ell-1}$, $p$ takes the same step as in $\beta$.  
    \item $y_p^\ell = m $. In this case, according to the implementation of halving objects (Algorithm~\ref{alg:halving}, Line~\ref{l:halve:new})  $p$ writes to $R_0$ and does not read register $R_1$ in $\beta$. We can therefore extend $\beta'_{\ell}$ with the steps taken by $p$ in its iteration $\ell$ in $\beta$. 
    \item $y_p^\ell > m$. In $\beta$, $p$ writes to $R_1$ before reading $R_0$. If it reads $R_0 = false$ in $\beta$, which is also the value of $R_0$ at the end of $\beta'_{\ell-1}$,  we let $p$ perform iteration $\ell$ after $\beta'_{\ell}$. 
    If $p$ reads \emph{true} from $R_0$ in $\beta$, we first have $q_\ell$ perform its next step. Recall that $q_\ell$ is about to perform $H_{S^\ell}.\mathsf{halve}(x_\ell)$ with $x_\ell \in S_\ell \setminus S_{\ell+1}$. Therefore, $x_\ell$ lies in the lower half of $S_\ell$ and the next step of $q_\ell$ is writing \emph{true} to $R_0$.  We then let $p$ perform its $\ell$th iteration. \qedhere
    \end{itemize}
\end{proof}

Let $aop$ be an invocation by process $p$ with input $x$ that returns $v$, 
and $aop'$ be another invocation by process $p'$ with input $x'$. 
By Lemma~\ref{lem:sch:ind}
there is an execution $\beta' \sim_p \beta$ in which $aop'$ does not appear. 
As the algorithm is wait-free, and solves approximate agreement (by Lemma~\ref{lem:halving_correct}), we have:

\begin{theorem}
  \label{prop:appx:correct_and_nl}
   Algorithm~\ref{alg:sch} together with Algorithm~\ref{alg:halving} are a wait-free, 
   $\APin$-non-leaking implementation of approximate agreement with a known range, using only registers.
\end{theorem}

The step complexity is $O(\log(1/\epsilon)$, and the space complexity is $O(\log^2(1/\epsilon))$. As observed by  Schenk,  space can be reduced to $O(\log(1/\epsilon))$  by dynamically assigning sub-intervals to halving objects, as only at most two such objects are used in each layer.



\section{Conclusions}

We introduced a framework for reasoning about information exposure in concurrent objects, using epistemic logic in labeled transition systems to compare what a process may know under an abstract specification with what it may know in a concrete implementation. 
A key novelty in our framework is the use of the abstract object itself as the reference for permissible information exposure.
Epistemic reasoning allows us to easily express different levels of information-hiding guarantees, 
while the use of labeled transition systems smoothly extends the framework to decision tasks. 

We showed that strong information-hiding guarantees are achievable for several concurrent objects while preserving linearizability 
and wait-freedom. 
At the same time, our impossibility for unbounded max registers shows that the strongest form of information hiding is not always achievable without compromise. 
Together, these results illustrate both the possibilities and the limitations of preserving information-hiding guarantees in concurrent objects.

Our work leaves open the design of additional non-leaking concurrent objects and their complexity, 
as well as other general approaches for doing so (like our value splitting technique), 
and the characterization of which concurrent objects admit non-leaking implementations. 
It would also be interesting to 
better understand how information-hiding guarantees can be preserved across layers of abstraction.

\bibliographystyle{plainurl}
\bibliography{references}


\appendix
\section{Wait-Free Unbounded Max Register from Responseless CAS}
\label{app:max-reg-cas}

The algorithm utilizes a \emph{responseless} CAS register $R$ (i.e., a CAS that does not return \emph{success} or \emph{failure}).
To write a new value $x$,
a process simply executes $R.\mathsf{CAS}(y,x)$ for each $y = 0,\ldots,x-1$ (in this order).
To read the max-register, a process reads $R$ and returns its value.
The code of the algorithm appears in Algorithm~\ref{alg:maxregcas}.

We prove that the algorithm is linearizable by assigning a linearization point to each operation. 
We say that a $R.\mathsf{CAS}(y,x)$ in the execution is \emph{successful} if it changes the value of $R$ from $y$ to $x$.
Let $wop$ be a \textsc{writeMax} operation. 
If $wop$ performs a successful CAS, we call the operation \emph{visible} and assign its linearization point to this successful step.
Otherwise, its linearization point is assigned to the last CAS executed in the loop, specifically the failed $R.\mathsf{CAS}(x, x-1)$ where $x$ is the argument passed to $wop$.
For a \textsc{readMax} operation $rop$, the linearization point is exactly its single execution step, which reads $R$.

Because each operation is linearized at a point between its invocation and response, this linearization trivially respects the real-time order. Furthermore, since a visible \textsc{writeMax} operation replaces a smaller value previously written by a different visible \textsc{writeMax} operation (or the initial value), the ordering of the visible \textsc{writeMax} operations and \textsc{readMax} operations clearly satisfies the sequential specification.

It remains only to show that linearizing the invisible \textsc{writeMax} operations also satisfies the sequential specification. This is established by the following proposition, which shows that a visible \textsc{writeMax} operation with an argument greater than or equal to the argument of the invisible \textsc{writeMax} must have been linearized beforehand.
It follows from the fact that $R$'s value is strictly increasing and that \textsc{writeMax} operations attempt to update $R$ in increasing order. 

\begin{proposition}
    If a \textsc{writeMax}($x$) operation performs an unsuccessful $R.\mathsf{CAS}(y,x)$ for $y < x$, then at that exact point, the value of $R$ must be strictly greater than $y$. 
\end{proposition}

Finally, we prove that the algorithm is fully non-leaking:

\begin{lemma}
Algorithm~\ref{alg:maxregcas} is a fully-non-leaking implementation of an unbounded max register.  
\end{lemma}

\begin{proof}
  Let $\alpha,\beta$ be executions of respectively the abstract register and Algorithm~\ref{alg:maxregcas} such that $\alpha$ is the linearization of $\beta$, as defined above.
  Let $\alpha' \sim_p \alpha$. 
  We show that there exists a concrete execution $\beta' \sim_p \beta$ that contains exactly the operations in $\min_p(\alpha')$.
  By Lemma~\ref{lem:fnl reduction}, this implies that the algorithm is fully-non-leaking.

  We construct $\beta'$ by iterating over the operations of $p$ in $\min_p(\alpha')$, so that $\beta'$ contains only the operations from $\min_p(\alpha')$ that appear before the operation by $p$.
  Since $\min_p(\alpha') \sim_p \alpha$ and $\alpha$ is the linearization of $\beta$, both $\min_p(\alpha')$ and $\beta$ contain the exact same operations by process $p$, appearing in the exact same order. 

    If the next operation of $p$ is a \textsc{writeMax} operation, we simply append it to $\beta'$. Since each $\mathsf{CAS}$ returns no response, process $p$ cannot distinguish between these executions, despite the fact that each $\mathsf{CAS}$ attempt might yield a different outcome and modify the memory differently.

     If the next operation of $p$ is a \textsc{readMax} operation returning $x$,
    then since $\min_p(\alpha')$ respects the sequential specification, 
    there is a preceding \textsc{writeMax}$(x)$ in $\min_p(\alpha')$, and no preceding \textsc{writeMax} in $\min_p(\alpha')$ has argument $y > x$. 
    By construction, the latter also holds in $\beta'$.
    
    We consider the \textsc{writeMax} operation with the maximum argument (if it exists) in $\beta'$. If the argument for this \textsc{writeMax} is $x$, we simply append the \textsc{readMax} operation by $p$ to $\beta'$.
    Otherwise, there must be some \textsc{writeMax}$(x)$ operation by process $q \neq p$ that precedes the \textsc{readMax} operation in $\min_p(\alpha')$.
    We first append a complete \textsc{writeMax}$(x)$ by process $q$, and then append the \textsc{readMax} operation by $p$ to $\beta'$.

    Since $x$ is the maximal argument of any \textsc{writeMax} operation in $\beta'$, the algorithm guarantees that $R$ holds the value $x$ when $p$ reads it during its \textsc{readMax} operation.

    By construction, $\beta' \sim_p \beta$. In addition, due to the minimality of $\min_p(\alpha')$, any operation not executed by $p$ exists solely to justify a value returned by a \textsc{readMax} operation by $p$. Hence, $\min_p(\alpha')$ does not contain any operation that is not included in $\beta'$.
\end{proof}

 \begin{algorithm}[bt]
  \caption{Fully-non-leaking wait-free unbounded max register from responseless CAS}
  \label{alg:maxregcas}
  \begin{algorithmic}[1]
  \small
    \State\textbf{Shared objects}
    \State\hspace{\algorithmicindent}$R$: responseless CAS object, initialized to $0$
    \Statex
    \begin{minipage}[t]{0.45\textwidth}
    \Function{writeMax}{$x$}
        \For{$y$ from $0$ upto $x-1$}
            \State $R.\mathsf{CAS}(y,x)$
        \EndFor
    \EndFunction
    \end{minipage}
    \begin{minipage}[t]{0.45\textwidth}
    \Function{readMax}{$~$}
      \State \Return $R.\mathsf{read}()$
      \EndFunction
      \end{minipage}
  \end{algorithmic}
\end{algorithm}

\end{document}